\documentclass[11pt]{amsart}

\usepackage{framed}
\usepackage{listings} 
\usepackage{amsmath,amssymb,amscd,amsthm,indentfirst,eufrak}
\usepackage{amsfonts}
\usepackage[all]{xy}
\usepackage{graphicx}
 \usepackage{color}
\usepackage{stmaryrd}
\usepackage{enumitem}

\usepackage[T1]{fontenc}

\usepackage[colorinlistoftodos,prependcaption]{todonotes}

\newtheorem{prop}{Proposition}[section]
\newtheorem{theorem}[prop]{Theorem}

\newtheorem{lemma}[prop]{Lemma}
\newtheorem*{theorem*}{Theorem}

\theoremstyle{definition}
\newtheorem{defn}[prop]{Definition}
\newtheorem{example}[prop]{Example}

\def\vp{\varphi}

\def\F{\mathcal{F}}

\def\L{\ensuremath{\mathcal{L}}}
\def\M{\ensuremath{\mathcal{M}}}

\def\RR{\ensuremath{\mathbb{R}}}

\def\proj{\ensuremath{\mathrm{proj}}}

\def\Mborel{{\C M}^+}

\newcommand{\xto}[1]{\xrightarrow{#1}}
\newcommand{\OP}[1]{\operatorname{#1}}
\newcommand{\B}[1]{\mathbf{#1}}
\newcommand{\C}[1]{\mathcal{#1}}

\newcommand{\AL}[1]{{#1}}

\author{Jarek K\k{e}dra}
\address{University of Aberdeen and University of Szczecin}
\email{kedra@abdn.ac.uk}
\author{Assaf Libman}
\address{University of Aberdeen}
\email{a.libman@abdn.ac.uk}
\author{Victoria Steblovskaya}
\address{Bentley University}
\email{vsteblovskay@bentley.edu}

\begin{document}

\title{European baskets in discrete-time continuous-binomial market models}

\begin{abstract}
We consider a discrete-time incomplete multi-asset market model with continuous price jumps.
For a wide class of contingent claims, including European basket call options, we compute the bounds of the interval containing the no-arbitrage prices.
We prove that the lower bound coincides, in fact, with Jensen's bound.
The upper bound can be computed by restricting to a binomial model for which an explicit expression for the bound is known by an earlier work of the authors.
We describe explicitly a maximal hedging strategy  which is the best possible in the sense that its value is equal to the upper bound of the price interval of the claim.
Our results show that for any $c$ in the interval of the non-arbitrage contingent claim price at time $0$, one can change the boundaries of the price jumps to obtain a model in which $c$ is the upper bound at time $0$ of this interval. The lower bound of this interval remains unaffected.
\end{abstract}

\maketitle

\section{The main results}
\label{Sec:main results}

{\bf Discrete time continuous-binomial market model.}
We consider a {\em discrete-time} market model, see e.g \cite[\S 4.5]{MR3931730} or \cite[\S 3]{pliska}, with $m$ risky assets $S_1,\dots,S_m$ and a bond $S_0$ whose prices at time $k=0,\dots,n$ denoted $S_i(k)$, are random processes described as follows.
We fix the initial values $S_i(0)$ of the assets $(i=0,1,\ldots, m)$ and parameters $R>0$ (the interest rate) and $0<D_i< R< U_i$.
For time $k=1,2,\dots, n$ the random process is defined by
\begin{itemize}
\item $S_0(k) = R^kS_0(0)$, and
\item $S_i(k) = \Psi_i(k) S_i(k-1)$, for $i=1,2,\ldots,m$, where $\Psi_i(k)$ are random variables with values in $[D_i, U_i]$.
\end{itemize}
We call $\Psi_i(k)$ the {\em price jumps} at time $k$.
We emphasize that the price jumps are not assumed to be independent of each other nor identically distributed.

A {\em European basket call option} is a contingent claim with pay-off given by
\begin{equation}\label{E:European option F}
F = \left(\sum_{i=0}^m c_i \cdot S_i(n) -K\right)^+,
\end{equation}
where $K>0$ and $c_i\geq 0$ for $i=1,2,\ldots,m$ and $x^+:=\max\{x,0\}$.
Notice that there is no assumption on $c_0$.
The {\em rational values} of $F$ at time $k$ are the possible market values of the option at time $k$ so that no arbitrage occurs.
They are known to form an open interval $(\Gamma_{\min}(F,k),\Gamma_{\max}(F,k))$, where
$\Gamma_{\min}(F,k)$ and $\Gamma_{\max}(F,k)$ depend on the ``state of the world'' at time $k$, namely the
history of the market up to time $k$, and in particular they depend on the (current) values of the assets 
$S_i$ at time $k$.

The main result of this paper is the computation of $\Gamma_{\min}(F,k)$ and $\Gamma_{\max}(F,k)$.
It extends the main results of the authors' previous work \cite{KLS1} in which we consider discrete time {\em binomial} models, i.e ones in which the price jumps $\Psi_i(k)$ take values in (the discrete) set $\{D_i,U_i\}$ rather than the entire interval $[D_i,U_i]$.

\noindent
\subsection{Computation of $\Gamma_{\min}(F,k)$ and $\Gamma_{\max}(F,k)$}
For every $1 \leq i \leq m$ set
\begin{equation}\label{Ebi from DU}
b_i=\frac{R-D_i}{U_i-D_i}
\end{equation}
and reorder the assets $S_1,\dots,S_m$, if necessary, so that 
\[
b_1 \geq \dots \geq b_m.
\]
Define $q_1,\dots,q_m$ by
\begin{equation}\label{E:def qi for main theorem}
q_i=\left\{
\begin{array}{ll}
1-b_1  & \text{if $i=1$} \\
b_i-b_{i+1} & \text{if $1<i<m$} \\
b_m        & \text{if $i=m$}
\end{array}\right.
\end{equation}
For any $0 \leq i \leq m$ and any $0 \leq j \leq m$ define numbers $\chi_i(j)$ as follows.
\begin{eqnarray}\label{E:def q_i}
&& \chi_i(j) = \left\{
\begin{array}{ll}
U_i    & \text{if $i \leq j$} \\
D_i    & \text{if $i>j$}
\end{array}\right.
\qquad \text{ for $1 \leq i \leq m$, and}
\\
\nonumber
&& \chi_0(j)=R.
\end{eqnarray}
To streamline the notation, set for every $k \geq 0$
\begin{equation}\label{E:def Pkm}
\C P_k(m)=\{0,\dots,m\}^k.
\end{equation}
Thus, any $J \in \C P_k(m)$ is a sequence $J=(j_1,\dots,j_k)$ with $0 \leq j_1,\dots,j_k \leq m$.
For such $J$ set
\begin{eqnarray}\label{E:def qJ and chiJ}
q_J &=& \prod_{j \in J} q_j \\
\nonumber
\chi_i(J) &=& \prod_{j \in J} \chi_i(j).
\end{eqnarray}

\begin{theorem}\label{T:rational values interval ends}
With the setup of the market model and notation above, the extremal values of the rational values of $F$ at time $0 \leq k \leq n$ are given by
\begin{eqnarray*}
\Gamma_{\min}(F,k) &=& R^{k-n} \cdot \left(R^{n-k} \sum_{i=0}^m c_i \cdot S_i(k) - K\right)^+ \\
\Gamma_{\max}(F,k) &=& 
R^{k-n} \cdot \sum_{J \in \C P_{n-k}(m)} q_J \cdot \left(\sum_{i=0}^m c_i \cdot \chi_i(J) \cdot S_i(k)  - K\right)^+.
\end{eqnarray*}
\end{theorem}

\subsection{Hedging strategies}
Consider a sequence of (time changing) portfolios
\[
V_\alpha(k) = \sum_i \alpha_i(k) S_i(k), \qquad (0 \leq k \leq n-1) 
\]
for some choices of values for $\alpha_i(k)$ at time $0 \leq k \leq n-1$.
A {\em maximum hedging strategy} is a choice for $\alpha_i(k)$ at time $0 \leq k \leq n-1$ (which depends on the state of the world at that time) which minimizes the value $V_\alpha(k)$ subject to the requirement that
\[
\sum_{i=0}^m \alpha_i(k) \cdot S_i(k+1)  \geq \Gamma_{\max}(F,k+1)
\]
for any subsequent state of the world at time $k+1$.
That is, a maximum hedging strategy is a time dependant portfolio of minimum possible cost whose value is guaranteed to exceed the future rational value of the contingent claim $F$.

Our next result, Theorem \ref{T:maximal hedging}, shows that the value of any hedging portfolio $V_\beta(k)$ must always exceed $\Gamma_{\max}(F,k)$ and that there exists a hedging strategy $\alpha_i(k)$ that attains this bound.
It is a {\em minimum cost maximum hedging strategy}.
To make this result precise, given the state of the world at time $0 \leq k \leq n-1$, let $Y_i(k)$ where $0 \leq i \leq m$ be the value of $\Gamma_{\max}(F,k+1)$ at the state of the world at time $k+1$ which is obtained from the present one (at time $k$) by having the assets $S_1,\dots,S_i$ make their maximum price jumps $U_1,\dots,U_i$ and having $S_{i+1},\dots,S_m$ make their minimum price jumps $D_{i+1},\dots,D_m$.
Explicitly, for any $0 \leq t \leq m$:
\[
Y_t(k) = R^{k+1-n} \cdot \sum_{J \in \C P_{n-k-1}(m)} q_{J} \left(\sum_{i=0}^m c_i \cdot \chi_i(J) \chi_i(t) \cdot S_i(k) -K \right)^+.
\]

\begin{theorem}\label{T:maximal hedging}
Consider the continuous binomial market model above and a European option $F$. 
Any hedging strategy $\beta_i(k)$ satisfies
\[
V_\beta(k) \geq \Gamma_{\max}(F,k).
\]
There exists a maximum hedging strategy $\alpha_i(0), \dots, \alpha_i(n-1)$ such that $V_\alpha(k)=\Gamma_{\max}(F,k)$ for all $0 \leq k \leq n-1$.
In fact, the values of $\alpha_i(k)$ at time $k$ are computed as follows.
\[
\begin{bmatrix}
\alpha_0(k) \\
\alpha_1(k) \\
\vdots \\
\alpha_m(k)
\end{bmatrix}
= 
W(k) \cdot N \cdot Q \cdot 
\begin{bmatrix}
Y_0(k) \\
Y_1(k) \\
\vdots \\
Y_m(k) 
\end{bmatrix}
\]
Where $W(k),N,T$ are the following $(m+1) \times (m+1)$ matrices.
Set $\Delta_i=U_i-D_i$.
\begin{eqnarray*}
&& W(k)= 
\begin{bmatrix}
\tfrac{1}{R \cdot S_0(k)}    \\
                             & \tfrac{1}{{S_1}(k)} \\
                             &                     & \ddots \\
                             &                     &        & \tfrac{1}{S_m(k)}
\end{bmatrix}
\\
&& Q=
\begin{bmatrix}
1  &  0 & 0 & \cdots \cdots & 0 & 0 \\
-1 &  1 & 0 & \cdots \cdots & 0 & 0 \\
0  & -1 & 1 & \cdots \cdots & 0 & 0 \\
\vdots & \vdots & \vdots & \cdots \cdots &  & \vdots \\
0 &   0 & 0 & \cdots \cdots & 1 & 0 \\
0 &   0 & 0 & \cdots \cdots & -1 & 1 \\
\end{bmatrix}
\\
&& N = 
\left[
\begin{array}{c|cccc}
1       & -\tfrac{D_1}{\Delta_1} & -\tfrac{D_2}{\Delta_2} & \cdots \cdots  &  -\tfrac{D_m}{\Delta_m} \\
\hline
0       & \tfrac{1}{\Delta_1}    &  0                    & \cdots \cdots  &  0 \\
0       & 0                      & \tfrac{1}{\Delta_2}   & \cdots \cdots  &  0 \\
\vdots  & \vdots                 & \vdots                & \vdots \\
0       & 0         
             & 0                     &  \cdots \cdots &  \tfrac{1}{\Delta_m}
\end{array}
\right]
\end{eqnarray*}
(Notice that $W(k)$ depends on the state of the world, but $N$ and $Q$ do not.
Also, $R \cdot S_0(k) = S_0(0) \cdot R^{k+1}$).
\end{theorem}

\AL{This result extends our previous result in \cite{KLS3} which computes a similar hedging strategy in discrete time binomial models, i.e models in which $\Psi_i(k) \in \{D_i,U_i\} \subset [D_i,U_i]$.}

\noindent
{\bf Changing the parameters of the model.}
Keeping $R$ fixed, we may change the values of $U_i$ and $D_i$ to obtain different models for the same market. 
This has the effect of changing the limits of the price jumps of the assets $S_i$, and consequently the random processes $S_i$ are changed.
Clearly the values of $\Gamma_{\min}(F,k)$ and $\Gamma_{\max}(F,k)$ depend on these parameters, and we therefore write $\Gamma_{\min}(F,k;U_i,D_i)$ and $\Gamma_{\max}(F,k;U_i,D_i)$ to emphasise this dependence.
We will be interested in the rational prices of $F$ at time $0$, namely $\Gamma_{\min}(F,0;U_i,D_i)$ and $\Gamma_{\max}(F,0;U_i,D_i)$.

\begin{theorem}\label{T:shrinking ends}
Consider a market model with some $0<D_i<R<U_i$ and the European basket $F$ in \eqref{E:European option F}.
Then

\begin{enumerate}
\item 
$\Gamma_{\min}(F,0;U_i,D_i)=(R^n \sum_{i=0}^m c_iS_i(0) -K)^+$ and in particular it is independent of the values of $U_i,D_i$.

\item
For every $c$ in the open interval $( \, \Gamma_{\min}(F,0;U_i,D_i) \,,\, \Gamma_{\max}(F,0;U_i,D_i)\,)$ there exist $D_i \leq d_i < R < u_i \leq U_i$ such that $\Gamma_{\max}(F,0;u_i,d_i)=c$.

More precisely, consider the functions $u_i, d_i \colon [0,1) \to \RR$ defined by
\begin{eqnarray*}
d_i(s) &=& D_i+(R-D_i)s \\
u_i(s) &=& \frac{R-(1-b_i)d_i(s)}{b_i}
\end{eqnarray*}
Then $\varphi(s) = \Gamma_{\max}(F,0;u_i(s),d_i(s))$ is a continuous function of $s \in [0,1)$ such that $\vp(0)=\Gamma_{\max}(F,0;U_i,D_i)$ and $\lim_{s \nearrow 1} \varphi(s) = \Gamma_{\min}(F,0;U_i,D_i)$.
\end{enumerate}
\end{theorem}

\section{Preliminaries: Random processes and conditional expectation}
\label{Sec:probability theory}

\subsection{Non-degenerate density functions}
This section is concerned with some general results about probability measure spaces.
Our standard reference for Measure theory and Lebesgue integration are Halmos \cite{MR0033869} and Royden \cite{MR1013117}, and for Probability Theory it is Feller \cite{MR0038583}.

Throughout this paper, once $m \geq 1$ is fixed we will denote
\begin{equation}
\label{E:def Omega}
\Omega = [0,1]^m \subseteq \RR^m
\end{equation}
equipped with the usual Borel $\sigma$-algebra and probability measure $\mu$.

A {\em probability density function (pdf)} is a measurable $p \colon \Omega \to [0,\infty)$ such that $\int_{\Omega} p d\mu =1$.
It gives rise in a standard way to a probability measure on $\Omega$ which by abuse of notation we also denote by $p$.

A probability measure on $\Omega$ is called {\em absolutely continuous} with respect to $\mu$, written $\nu \ll \mu$, if for any Borel subset $E$ we have $\mu(E)=0 \implies \nu(E)=0$.
A probability measure $\nu$ on $\Omega$ is {\em non-degenerate} if $\nu \ll \mu$ and $\mu \ll \nu$.
We write $\nu \approx \mu$.

By the Radon-Nykodim theorem \cite[\S 11.5] {MR1013117} if $\nu \ll \mu$ then there exists a pdf $p \colon \Omega \to [0,\infty)$ called the Radon-Nykodim derivative, such that $\nu(E)=\int_{E} p(x)\, d\mu(x)$.
It is easy to check that $\nu \approx \mu$ if and only if $p>0$ almost everywhere\footnote{If $p>0$ a.e  then  $\int_E p\,du>0$ for any $E$ with $\mu(E)>0$ is a standard fact. If $p=0$ on $E$ with $\mu(E)=0$ then $\nu(E)=\int_E p(x)\, d\mu(x) = 0$ so $\mu$ is not absolutely continuous with respect to $\nu$.}.
We say that $p$ is non-degenerate.

\subsection{Conditional probability}
\label{SS:conditional expectation}


Consider $\Omega_{(i)} = [0,1]^{m_i}$ where $i=1,\dots,n$.
Set 
\[
\Omega = \prod_{i=1}^n \Omega_{(i)} = [0,1]^m,
\]
where $m=\sum_i m_i$, equipped with the standard Borel $\sigma$-algebra.
Let $X_{(i)}$ denote the random vector
\[
X_{(i)} \colon \Omega \xto{\ \proj_i \ } \Omega_{(i)} \subseteq \RR^{m_i}.
\]
For a non-empty $I \subseteq \{1,\dots,n\}$ set $m_{I} = \sum_{i \in I} m_i$.
Then set
\begin{align*}
& \Omega_{(I)} = \prod_{i \in I} \Omega_{(i)}  \\
& X_{(I)} \colon \Omega \xto{\ \proj_{(I)} \ } \Omega_{(I)} \subseteq \RR^{m_{I}}.
\end{align*}
Thus, $\Omega_{(I)} = [0,1]^{m_{I}}$ and $X_{(I)}$ is a random vector into $\RR^{m_{I}}$.
We will denote elements of $\Omega_{(I)}$ by $\omega_{(I)}$.
We will write $n-I$ for the complement of $I$.

For the remainder of this subsection we fix a non-degenerate (with respect to the Lebesgue measure on $\Omega$) pdf  $p \colon \Omega \to [0,\infty)$ and equip $\Omega$ with the probability measure it induces which we abusively denote by $p$.
Note that $p$ is the joint density function of the random vectors $X_{(1)}, \dots, X_{(n)}$ (because $X_{\{1,\dots,n\}}$ is the inclusion $\Omega \subseteq \RR^m$).

Given a non-empty $I \subseteq \{1,\dots,n\}$, the (joint) density function of $X_{(I)}$ is the function $p_{X_{(I)}} \colon \Omega_{(I)} \to [0,\infty)$ given by (See e.g. \cite[Chap. 2, Scet. 3]{MR2977961}) 
\begin{equation}\label{E:def pXI}
p_{X_{(I)}}(\omega_{(I)}) = \int_{\tau \in \Omega_{(n-I)}} p(\omega_{(I)},\tau) d\tau.
\end{equation}
Fubini's theorem readily implies that $p_{X_{(I)}}$ is non-degenerate (with respect to the Lebesgue measure on $\Omega_{(I)}$).

Consider some {\em disjoint} $I,J \subseteq \{1,\dots,n\}$.
The density function of $X_{(I)}$ given $X_{(J)}$ denoted $p_{X_{(I)}|X_{(J)}} \colon \Omega_{(I \cup J)} \to [0,\infty)$ is 
\begin{equation}\label{E:def pXI|XJ}
p_{X_{(I)}|X_{(J)}} \,(\omega_{(I)},\omega_{(J)}) 
= \frac{p_{X_{(I\cup J)}}(\omega_{(I)},\omega_{(J)})}{p_{X_{(J)}(\omega_{(J)})}}
\end{equation}
whenever this is defined.
Since $p_{X_{(J)}}$ is non-degenerate, $p_{X_{(I)}|X_{(J)}}$ is defined a.e.
For any $\omega_{(J)}$ we obtain a function, the (conditional) density of $X_{(I)}$ given the event $\{X_{(J)}=\omega_{(J)}\}$,
\[
p_{X_{(I)}|X_{(J)}=\omega_{(J)}} \colon \Omega_{(I)} \to [0,\infty)
\]
defined by $p_{X_{(I)}|X_{(J)}=\omega_{(J)}}(-) = p_{X_{(I)}|X_{(J)}}(-, \omega_{(J)})$.
By Fubini's theorem it is a (measurable) probability density function on $\Omega_{(I)}$.

Consider a random vector
\[
f \colon \Omega \to \RR^k.
\]
To avoid issues of convergence we  assume that $f \geq 0$, namely all the components of $f$ are non-negative.
The {\em expectation of $f$ given $X_{(I)}$} is the function
\[
E_p(f|X_{(I)}) \colon \Omega_{(I)} \to \RR
\]
where $E_p(f|X_{(I)})(\omega_{(I)})$ is the conditional expectation $E_p(f|X_{(I)}=\omega_{(I)})$, namely 
\[
E_p(f|X_{(I)})(\omega_{(I)}) = \int_{\tau \in \Omega_{(n-I)}} f(\tau,\omega_{(I)}) p_{X_{(n-I)}|X_{(I)}}(\tau,\omega_{(I)}) d\tau.
\]
By Fubini's theorem $E_p(f|X_{(I)})$ is a measurable function.

\begin{lemma}\label{L:conditional expectation in steps}
Keeping the notation above, let $I,J \subseteq \{1,\dots,n\}$ be disjoint and consider some $\omega_{(I)} \in \Omega_{(I)}$.
Set $p'=p_{X_{(J)}|X_{(I)}=\omega_{(I)}}$, a pdf on $\Omega_{(J)}$.
Let $f \geq 0$ be a random vector on $\Omega$.
Then
\[
E_p(f|X_{(I)}=\omega_{(I)}) \, = \, E_{p'}(\tau \mapsto E_p(f|X_{(I\cup J)})(\omega_{(I)},\tau)).
\]
In particular, 
\[
E_p(X_{(J)}|X_{(I)}=\omega_{(I)}) \, = \, E_{p'}(X_{(J)})
\]
where $X_{(J)}$ is viewed as a random vector from $\Omega_{(J)}$.
\end{lemma}

\begin{proof}
We compute the right hand side using Fubini's theorem as follows
\begin{align*}
E_{p'}(\tau \mapsto & E_p(f|X_{(I\cup J)})(\omega_{(I)},\tau)) = 
\\
& =
\int_{\tau \in \Omega_{(J)}} p_{X_{(J)}|X_{(I)}=\omega_{(I)}}(\tau) \cdot E_p(f|X_{(I \cup J)})(\omega_{(I)},\tau)\, d\tau 
\\
&=
\int_{\tau \in \Omega_{(J)}} \left( \tfrac{p_{X_{(I \cup J)}}(\omega_{(I)},\tau)}{p_{X_{(I)}}(\omega_{(I)})} \cdot \int_{\theta \in \Omega_{(n-I \cup J)}} f(\omega_{(I)},\tau,\theta) \cdot \tfrac{p(\omega_{(I)},\tau,\theta)}{p_{X_{(I \cup J)}}(\omega_{(I)},\tau)} \, d\theta \right) \, d\tau 
\\
& =
\int_{\tau \in \Omega_{(J)}} \int_{\theta \in \Omega_{(n-I \cup J)}} f(\omega_{(I)},\tau,\theta) \cdot \tfrac{p(\omega_{(I)},\tau,\theta)}{p_{X_{(I)}}(\omega_{(I)})} d\theta \, d\tau 
\\
&= 
\int_{\omega \in \Omega_{(n-I)}} f(\omega_{(I)},\omega) \cdot \tfrac{p(\omega_{(I)},\omega)}{p_{X_{(I)}}(\omega_{(I)})} \, d\omega 
\\
& = 
\int_{\omega \in \Omega_{(n-I)}} f(\omega_{(I)},\omega) \cdot p_{X_{(n-I)}|X_{(I)}=\omega_{(I)}}(\omega) \, d\omega 
\\
& =
E_p(f|X_{(I)}=\omega_{(I)}) 
\end{align*}
This establishes the first claim.
We apply it to the random vector $f=X_{(J)}$ to obtain the second claim as follows
\[
E_p(X_{(J)}|X_{(I)}=\omega_{(I)}) = 
E_{p'}(\tau \mapsto E_p(X_{(J)}|X_{(I \cup J)})(\omega_{(I)},\tau)) =
E_{p'}(\tau \mapsto X_{(J)}(\tau)) = E_{p'}(X_{(J)}).
\]
Here we observe that $E_p(X_{(J)}|X_{(I \cup J)})(\omega_{(I)},\tau)=X_{(J)}(\tau)$ because with the abuse of notation for the domain of $X_{(J)}$ we have $X_{(J)}(\omega_{(I)},\tau,\theta)=X_{(J)}(\tau)$ for any $\tau \in \Omega_{(J)}$ and any $\theta \in \Omega_{(n-I \cup J)}$.
\end{proof}

\begin{lemma}\label{L:conditional expectation of truncated functions}
Let $f \colon \Omega \to \RR^d$ be a function and $I,J \subseteq \{1,\dots,n\}$ be disjoint, and assume that $f \geq 0$.
Suppose that $f$ factors through the projection $\Omega \xto{\pi_{I \cup J}} \Omega_{(I \cup J)}$, namely there exists $g \colon \Omega_{(I \cup J)} \to \RR^m$ such that $f=g \circ \pi_{I \cup J}$.
Set $p'=p_{X_{(J)}|X_{(I)}=\omega_{(I)}}$, pdf on $\Omega_{(J)}$.
Then
\[
E_p(f|X_{(I)}=\omega_{(I)}) = E_{p'}(\omega_{(J)} \mapsto g(\omega_{(I)}\,\omega_{(J)}))
\]
\end{lemma}

\begin{proof}
Lemma \ref{L:conditional expectation in steps} gives
\[
E_p(f|X_{(I)}=\omega_{(I)})\,=\,E_{p'}(\omega_{(J)} \mapsto E_p(f|X_{(I \cup J)})(\omega_{(I)},\omega_{(J)})) =
E_{p'}(\omega^J \mapsto g(\omega^I,\omega^J))
\]
because $f(\omega_{(I)},\omega_{(J)},\tau)=g(\omega_{(I)},\omega_{(J)})$ for any $\tau \in \Omega_{(n-I\cup J)}$,  so $E_p(f|X_{(I \cup J)})(\omega_{(I)},\omega_{(J)}) = g(\omega_{(I)},\omega_{(J)})$.
\end{proof}

\subsection{Tensoring}
Given $p \colon \Omega \to \RR$ and $q \colon \Omega' \to \RR$ we obtain a function $p \otimes q \colon \Omega \times \Omega' \to \RR$ by
\begin{equation}\label{E:def tensor}
(p \otimes q)(\omega,\omega') = p(\omega) \cdot q(\omega').
\end{equation}
It is clear that if $p,q$ are pdf's then so is $p \otimes q$ and that it is non-degenerate if $p$ and $q$ are non-degenerate.

Keeping the notation above for $\Omega_{(i)}$ and the random vectors $X_{(i)}$, let $p_{(i)} \colon \Omega_{(i)} \to [0,\infty)$ be non-degenerate pdf's.
Then $p=p_{(1)} \otimes \cdots \otimes p_{(n)}$ is a non-degenerate pdf on $\Omega=\prod_i \Omega_{(i)}$.
For any $I \subseteq \{1,\dots,n\}$ we denote 
\[
p_{(I)}=\underset{i \in I}{\otimes} p_{(i)}. 
\]
This is a non-degenerate pdf on $\Omega_{(I)}$.

\begin{lemma}\label{L:tensors and conditional expectations}
Let $p_{(i)} \colon \Omega_{(i)} \to [0,\infty)$ be non-degenerate pdf's, $i=1,\dots,n$.
Set $p=p_{(1)} \otimes \dots \otimes p_{(n)}$, non-degenerate pdf on $\Omega$.
Then
\begin{enumerate}
\item \label{L:tensors p:q XI}
$p_{X_{(I)}} = p_{(I)}$ for any $I \subseteq \{1,\dots,n\}$.

\item \label{L:tensors p:q XI given XJ}
$p_{X_{(J)}|X_{(I)}=\omega_{(I)}} = p_{(J)}$ for any disjoint $I,J \subseteq \{1,\dots,n\}$, and furthermore

\item \label{L:tensors p:E XJ given XI}
$E_p(X_{(J)}|X_{(I)}=\omega_{(I)})= E_{p_{(J)}}(X_{(J)})$ where $X_{(J)}$ is viewed as a random vector on $\Omega_{(J)}$.
\end{enumerate}
\end{lemma}

\begin{proof}
(\ref{L:tensors p:q XI})
Given $\omega_{(I)}$ we compute
\[
p_{X_{(I)}}(\omega_{(I)})=\int_{\tau \in \Omega_{(n-I)}} p(\omega_{(I)},\tau) \, d\tau =
\int_{\tau \in \Omega_{(n-I)}} p_{(I)}(\omega_{(I)})\cdot p_{(n-I)}(\tau) \, d\tau =
p_{(I)}(\omega_{(I)})
\]
because $p_{(J)}$ is a pdf on $\Omega_{(J)}$ for any $J$.

(\ref{L:tensors p:q XI given XJ})
Given $\omega_{(J)}$ and $\omega_{(I)}$ we use this to compute
\begin{multline*}
p_{X_{(J)}|X_{(I)}=\omega^I} (\omega_{(J)}) 
= \frac{p_{X_{(I \cup J)}}(\omega_{(I)},\omega_{(J)})}{p_{X_{(I)}}(\omega_{(I)})} =
\frac{p_{(I \cup J)}(\omega_{(I)},\omega_{(J)})}{p_{(I)}(\omega_{(I)})} \\
=
\frac{p_{(I)}(\omega_{(I)}) \cdot p_{(J)}(\omega_{(J)})}{p_{(I)}(\omega_{(I)})} 
= p_{(J)}(\omega_{(J)}).
\end{multline*}
(\ref{L:tensors p:E XJ given XI}) By Lemma \ref{L:conditional expectation in steps} and item (\ref{L:tensors p:q XI given XJ})
\begin{multline*}
E_p(X_{(J)}|X_{(I)}=\omega_{(I)}) = 
E_{p_{X_{(J)}|X_{(I)}=\omega_{(I)}}}( \omega_{(J)} \mapsto E_p(X_{(J)}|X_{(I \cup J)}(\omega_{(I)},\omega_{(J)})) \\ 
=
E_{p_{(J)}}(\omega_{(J)} \mapsto X_{(J)}(\omega_{(J)})).
\end{multline*}
\end{proof}

\subsection{Processes}\label{SS:processes}

\begin{defn}
Let $\Omega$ be a set and $n \geq 1$.
An $\RR^d$-valued {\em process} (over $\Omega$) is a sequence of functions
\[
Y^{(0)},\dots,Y^{(n)} \colon \Omega^n \to \RR^d
\]
such that for each $0 \leq k \leq n$ the function $X^{(k)}$ factors through the projection $\Omega^n \to \Omega^k$ to the first $k$ factors.

If $\Omega^n$ is equipped with a probability measure, $Y^{(0)},\dots, Y^{(k)}$ is called an $\RR^d$-valued {\em random process}.
\end{defn}

We frequently regard $Y^{(k)}$ as a function with domain $\Omega^k$.
We will sometimes ``trim'' the process to $Y^{(1)},\dots,Y^{(n)}$ or to $X^{(0)},\dots,X^{(n-1)}$.

If $\Omega = [0,1]^m \subseteq \RR^m$ the projections $L^{(k)} \colon \Omega^n \to \Omega \subseteq \RR^m$ to the $k$-th factor give a universal process in the sense that if $Y^{(0)},\dots,Y^{(n)}$ is a process then each $Y^{(k)}$ is a function of $L^{(1)},\dots,L^{(k)}$.

\subsection{Supports}
Let $\nu$ be a probability measure on a set $\Omega$.

\begin{defn}
We say that $\nu$ is supported on a measurable set $A$ if $\mu(A)=1$.
\end{defn}

Any probability measure $\nu$ on $A \subseteq \Omega$ extends to a probability measure $\tilde{\nu}$ on $\Omega$ supported on $A$ by $\tilde{\nu}(E)=\nu(A \cap E)$.
Conversely, if $\nu$ on $\Omega$ is supported by $A$ then $\nu=\widetilde{\nu|_A}$.
If $A$ is finite then for any $f \colon \Omega \to \RR$ we have $E_\nu(f)=\sum_{a \in A} \nu(\{a\}) f(a)$.

\section{Maximum and minimum expectation of random variables on $\Omega=[0,1]^m$}
\label{SEC:Omega}

\noindent
Fix some $m \geq 1$ and $\Omega=[0,1]^m$ equipped with the usual Lebesgue measure $\mu$. 
Let 
\[
L \colon \Omega \to \RR^m
\] 
denote the inclusion.
We think of it as a random vector with components $L=(\ell_1, \dots, \ell_m)$.
Thus, $\ell_i \colon \Omega \to \RR$ is the projection to the $i$th factor
\[
\ell_i(x_1,\dots,x_m)=x_i.
\]
For convenience we also set 
\[
\ell_0=1,
\] 
the constant function (random variable).

\begin{defn}
Let $\C M(\Omega)$ denote the set of {\em all} probability measures on $\Omega$ on the Borel $\sigma$-algebra.
Let $\C M^+(\Omega)$ denote the set of all the non-degenerate probability measures (with respect to the Lebesgue measure).
\end{defn}

We will often refer to the elements of $\Mborel(\Omega)$ as pdf's which are non-vanishing a.e.

Consider a non decreasing  $b \in \OP{int}(\Omega)=(0,1)^m$, namely 
\[ 
1 > b_1 \geq \cdots \geq b_m > 0.
\] 
We will also denote for convenience
\[
b_0=1 \qquad \text{and} \qquad b_{m+1}=0.
\]

\begin{defn}
The set of {\em mean-$b$} probability measures on $\Omega$ is
\[
\C M(\Omega,b) = \{ P \in \C M(\Omega) : E_P(L)=b\}.
\]
That is, $E_P(\ell_i)=b_i$ for all $1 \leq i \leq m$.
The set of non-degenerate mean-$b$ probability measures is 
\[
\Mborel(\Omega,b) = \{ P \in \C \Mborel(\Omega) : E_P(L)=b\}.
\]
\end{defn}

\begin{defn}\label{D:vertices of cube L}
Let $\L$ denote the set of vertices of the cube $\Omega=[0,1]^m$.
That is,
\[
\L=\{0,1\}^m.
\]
\end{defn}

There is a standard identification of $\L$ with $\wp(\{1,\dots,m\})$ where $\lambda \in \L$ corresponds to $\OP{supp}(\lambda)$.
This turns $\L$ into a lattice where the partial order $\preceq$ is induced by inclusion of sets and joins and meets are $\cup$ and $\cap$.
The next concept is originally due to Lov\'{a}sz \cite{MR717403}.

\begin{defn}
A function $f \colon \L \to \RR$ is called {\em supermodular} if for any $a,b \in \L$ 
\[
f(a \vee b) + f(a \wedge b) \geq f(a)+f(b).
\]
It is called {\em modular} if equality holds.
\end{defn}

\begin{defn}\label{D:convex-supermodular}
A function $f \colon \Omega \to \RR$ is called {\em convex-supermodular} if $f$ is convex and its restriction to $\L$ is supermodular.
\end{defn}

\begin{example}\label{Ex:convex-supermodular}
Let $f \colon \Omega \to \RR$ and suppose that $f=h \circ g$ for some $g \colon [0,1]^m \to \RR$ and $h \colon \RR \to \RR$ such that either
\begin{enumerate}
\item\label{Ex:convex-supermodular:1}
$h$ is convex and $g$ is affine with non-negative coefficients except the constant term, namely $g=\sum_{i=0}^m a_i \ell_i$ where $a_1,\dots,a_m \geq 0$.

\item\label{Ex:convex-supermodular:2}
$g$ is convex, $g|_\L$ is modular, and $h$ is convex and increasing.
\end{enumerate}
Then $f$ is convex-supermodular.

\noindent
{\em Proof:} Indeed, $f$ is convex as composition of convex functions and $f|_\L$ is supermodular by [Simchi-Levi, Theorem 2.2.6] for item (\ref{Ex:convex-supermodular:1}) and [Simchi-Levi Proposition 2.2.5(c)] for item   (\ref{Ex:convex-supermodular:2}). 
\hfill $\diamondsuit$
\end{example}

For every $0 \leq k \leq m$ let $\rho_k \in \L$ denote the element corresponding to $\{1,\dots,k\}$, namely  
\begin{equation}\label{E:define rho_k}
\rho_k = (\underbrace{1,\dots,1}_{\text{$k$ times}},0,\dots,0) \in \{0,1\}^m.
\end{equation}

\begin{defn}[Compare \cite{KLS1}] \label{D:upper supermodular vertex}
The {\em upper supermodular vertex} is the probability density function $q^* \colon \L \to \RR$ supported on $\{\rho_0,\dots,\rho_m\}$ with
\[ 
q^*(\rho_k)=b_k-b_{k+1}, \qquad (0 \leq k \leq m).
\] 
\end{defn}
One easily checks that $\sum_{i=0}^m q^*(\rho_k)=1$ and that $E_{q^*}(\ell_i)=b_i$, thus

\begin{prop}
$q^* \in \C M(\Omega,b)$ and it is supported on $\L \subseteq \Omega$.
\end{prop}


The main result of this section is the following theorem.

\begin{theorem}\label{T:one step max}
Let $f \colon \Omega \to \RR$ be a continuous convex-supermodular function and assume that $f \geq 0$.
Let $q^*$ be the upper supermodular vertex of $\L$ (Definition \eqref{D:upper supermodular vertex}).
Then
\[
\sup_{p \in \Mborel(\Omega,b)} E_p(f) \ = \ E_{q^*}(f|_\L) \  = \ \max_{p \in \M(\Omega,b)} E_p(f). 
\]
Note that $E_{q^*}(f|_\L)= \sum_{i=0}^m q^*(\rho_i) \cdot f(\rho_i)$.
\end{theorem}

In the remainder of this section we prove this theorem.
It relies on the following key observation.
Equip $\RR^m$ with the norm $\| \ \|_{\infty}$ and restrict this norm to $\Omega=[0,1]^m$.

\begin{lemma}[Approximation lemma]
\label{L:approximation n=1}
Consider some $q \in \M(\Omega,b)$ supported on a finite subset $\{x^1,\dots,x^k\}$ of $\Omega$ and set $q_i=q(\{x^i\})$.
Suppose that $b \in \OP{int}(\Omega)=(0,1)^m$.
Then for any $\epsilon>0$ and $\delta>0$ there exists $p \in \Mborel(\Omega,b)$ and $\beta<\epsilon$ such that for any continuous $f \colon \Omega \to \RR$
\[
E_p(f) = \beta \int_\Omega f\, d\mu + \sum_{i=0}^k (q_i-\tfrac{\beta}{k}) \cdot f(\xi^i)
\]
where $\xi^1,\dots,\xi^k \in \Omega$ are such that $\|\xi^i-x^i\|_\infty<\delta$.
\end{lemma}

\begin{proof}
Closed balls of radius $r>0$ in $\RR^m$ have the form $B(y,r)= y+[-r,r]^m$ so their volume, hence their Lebesgue measure, is $(2r)^m$.
If $y=(y_1,\dots,y_m)$ then by inspection, for any $1 \leq j \leq m$
\[
\int_{B(y,r)} x_j\, d\mu(x^1,\dots,x^m) = (2r)^m y_j.
\]

\noindent
{\em Claim:} There exist distinct $y^1,\dots,y^k$ in $\OP{int}(\Omega)=(0,1)^m$ such that $\|y^i-x^i\|_\infty < \tfrac{\delta}{3}$ and such that $\sum_{i=1}^k q_i y^i=b$.

\noindent
{\em Proof:} We show how to perturb the vectors $x^1,\dots,x^k$ in order to obtain $y^1,\dots,y^k$.
First, $\sum_{i=1}^k q_i x^i=E_q(L)=b$ since $q \in \C M(\Omega,b)$.
Suppose that for some $1 \leq j \leq m$ not all of $x^1_j,\dots,x^k_j$ are in the open interval $(0,1)$.
If $x^{i'}_j=0$ for some $i'$ then there must exist $i''$ such that $x^{i''}_j >0$ because $\sum_{i=1}^k q_i x^i_j=b_j>0$.
Since $q_i >0$ for all $i$ we can increase $x^{i'}_j$ and decrease $x^{i''}_j$ by a small positive number $<\tfrac{\delta}{3}$ so that the sum remains $b_j$ and that the new values of $x^{i'}_j$ and $x^{i''}_j$ are in $(0,1)$.
Similarly, if $x^{i'}_j=1$ for some $i'$ then there must exist some $i''$ such that $x^{i''}_j<1$ because $\sum_{i=1}^k q_i x^i_j=b_j<1$.
We can then decrease $x^{i'}_j$ and increase  $x^{i''}_j$ by at most $\tfrac{\delta}{3}$ so that the sum remains $b_j$ and the new values of $x^{i'}_j$ and $x^{i''}_j$ are in $(0,1)$.
By repeating this process we can perturb $x^1,\dots,x^k$ into $y^1,\dots,y^k$ in $\OP{int}(\Omega)$ such that $\sum_i q_i y^i=b$ and $\|x^i-y^i\|_\infty<\tfrac{\delta}{3}$.
Since the $x^i$'s admit pairwise disjoint neighbourhoods, we can perturb the $x^i$'s inside these neighbourhoods to make sure that the $y^i$'s are distinct.
\hfill q.e.d

Since $y^1,\dots,y^k$ are in the interior of $\Omega$ there exists $r<\tfrac{\delta}{3}$ such that $B(y^i,2r) \subseteq (0,1)^m$ for all $1 \leq i \leq k$.
Thus, if $\gamma \in \RR^m$ is such that $\| \gamma \|_\infty<r$ then $B(y^i+\gamma,r) \subseteq (0,1)^m$.
Set 
\[
Q=\min\{ q_1,\dots,q_m\}.
\]
For any $\beta >0$ set
\[
\gamma^1_j := \frac{\beta \cdot (\tfrac{1}{k}\sum_{i=1}^k y^i_j - \tfrac{1}{2})}{q_1-\tfrac{\beta}{k}}.
\]
Choose  $0<\beta < \min\{ \epsilon, kQ\}$ sufficiently small such that for every $1 \leq j \leq m$
\[
|\gamma^1_j| < r.
\]
Let $\gamma^1 \in \RR^m$ be the vector with the components $\gamma^1_j$ defined above and let $\gamma^2,\dots,\gamma^k \in \RR^m$ be the zero vectors.
By construction $\|\gamma^i\|_\infty<r$ for all $1 \leq i \leq k$, hence $B(y^i+\gamma^i,r) \subseteq (0,1)^m$.
Define $p \colon \Omega \to \RR$ by
\[
p = \beta+\sum_{i=1}^k \tfrac{q_i-\tfrac{\beta}{k}}{(2r)^m}  \cdot \mathbf{1}_{B(y^i+\gamma^i,r)}
\]
where $\mathbf{1}_{B(y^i+\gamma^i,r)}$ is the characteristic function.
Observe that $p>0$ because $\beta>0$ and $q_i-\tfrac{\beta}{k} > q_i-Q \geq 0$.
Next, $p$ is a pdf since
\[
\int_{\Omega} p \, d\mu = \beta + \sum_{i=1}^k \frac{q_i-\tfrac{\beta}{k}}{(2r)^m}\int_{\Omega} \mathbf{1}_{B(y^i+\gamma^i,r)}  d\mu = \sum_{i=1}^k q_i = 1.
\]
We check that $p \in \Mborel(\Omega,b)$.
Indeed, since $\gamma^i=0$ for all $i \geq 2$
\begin{align*}
E_{p}(\ell_j) &= 
\int_{x \in \Omega} \ell_j(x) \cdot p(x) \, d\mu(x) 
\\
&=
\beta\int_\Omega x_j \, d\mu +\tfrac{1}{(2r)^m} \sum_{i=1}^k (q_i-\tfrac{\beta}{k}) \int_{B(y^i+\gamma^i,r)} x_j \, d\mu 
\\
&=
\tfrac{1}{2}\beta + \sum_{i=1}^k (q_i-\tfrac{\beta}{k})(y^i_j+\gamma^i_j) 
\\
&=
\tfrac{1}{2}\beta + \sum_{i=1}^k q_i y^i_j +  (q_1-\tfrac{\beta}{k})\gamma^1_{j}  -\tfrac{\beta}{k} \sum_{i=1}^k y^i_j 
\\
&= 
\tfrac{1}{2} \beta + b_j  + \beta (\tfrac{1}{k} \sum_{i=1}^k y_j^i - \tfrac{1}{2}) - \tfrac{\beta}{k} \sum_{i=1}^k y_j^i
\\
&= b_j.
\end{align*}
Suppose that $f \colon \Omega \to \RR$ is continuous.
By the mean value theorem there exist $\xi^i \in B(y^i+\gamma^i,r)$ such that
\begin{multline*}
E_{p}(f) = 
\int_{x \in \Omega} f(x) \cdot p(x) \, d\mu(x) =
\beta \int_\Omega f\, d\mu + \sum_{i=1}^k \tfrac{1}{(2r)^m} (q_i-\tfrac{\beta}{k}) \int_{B(y^i+\gamma^i,r)} f \, d\mu =
\\
\beta \int_{\Omega} f\, d\mu + \sum_{i=1}^k(q_i-\tfrac{\beta}{k}) \cdot f(\xi^i).
\end{multline*}
By our choice $\beta<\epsilon$ and $\|\xi^i-x^i\|_\infty \leq \|\xi^i-(y^i+\gamma^i)\|_\infty+\|y^i-x^i\|_\infty +\|\gamma^i\|_\infty < r+ \tfrac{\delta}{3} +r < \delta$.
This completes the proof.
\end{proof}


\begin{proof}[Proof of Theorem \ref{T:one step max}]
For $i=0,\dots,m$ set
\begin{eqnarray*}
&& \alpha_0=f(\rho_0) \\
&& \alpha_i = f(\rho_i)-f(\rho_{i-1}).
\end{eqnarray*}
Observe that since by definition $b_0=1$ and $b_{m+1}=0$,
\begin{align*}
E_{q^*}(f|_{\L}) &= 
\sum_{i=0}^m q^*(\rho_i) \cdot f(\rho_i) 
\\
&=
\sum_{i=0}^m (b_i-b_{i+1})f(\rho_i) 
\\
& = 
f(\rho_0) + \sum_{i=1}^m b_i(f(\rho_i)-f(\rho_{i-1})) 
\\
& =
\alpha_0+\sum_{i=1}^m b_i \alpha_i.
\end{align*}
Let $\check{f} \colon \RR^m \to \RR$ be the affine function $\check{f}=\sum_{i=0}^m \alpha_i \ell_i$, namely 
\[
\check{f}(x_1,\dots,x_m)=\alpha_0+\sum_{i=1}^m \alpha_i x_i.
\]

\noindent
{\em Claim 1:} $\check{f}$ dominates $f$ on $\L$, namely $\check{f}(\lambda) \geq f(\lambda)$ for all $\lambda \in \L$.

\noindent
{\em Proof:} By construction of $\check{f}$ and by the definition of $\rho_j$ in  \eqref{E:define rho_k}, for any $0 \leq j \leq m$
\[
\check{f}(\rho_j) = \sum_{i=0}^j \alpha_i = f(\rho_j).
\]
So $f$ and $\check{f}$ coincide on $\{\rho_0,\dots,\rho_m\} \subseteq \L$.
Assume the statement of the claim is false, namely there exists $\lambda \in \L$ such that  $\check{f}(\lambda) < f(\lambda)$.
Among all these $\lambda$'s choose one which contains the longest leading run of $1,\dots,1$, namely $\lambda$ with the largest possible $k$ with $\rho_k \preceq \lambda$.
Notice that $k<m$ because $\check{f}(\rho_m)=f(\rho_m)$ and $\rho_m$ is maximal in $\L$.
In the lattice $\L$ set $\lambda' = \lambda \vee \rho_{k+1}$.
Notice that since $k$ is the largest such that $\rho_k \preceq \lambda$ it follows that $\lambda \wedge \rho_{k+1} = \rho_k$.
Since $f|_\L$ is supermodular
\[
f(\lambda \vee \rho_{k+1}) + f(\rho_k) \geq f(\lambda)+f(\rho_{k+1}).
\]
Since $\check{f}$ is affine, it is modular, hence
\[
\check{f}(\lambda \vee \rho_{k+1}) + \check{f}(\rho_k) = \check{f}(\lambda)+\check{f}(\rho_{k+1}).
\]
Subtract the first equation from the second, taking into account that $\check{f}(\rho_i)=f(\rho_i)$, to get
\[
\check{f}(\lambda \vee \rho_{k+1}) - f(\lambda \vee \rho_{k+1}) \leq \check{f}(\lambda)-f(\lambda) <0.
\]
Thus $\check{f}(\lambda') < f(\lambda')$ and $\rho_{k+1} \preceq \lambda'$, contradiction to the maximality of $k$.
\hfill q.e.d 

\noindent
{\em Claim 2:} $\check{f}$ dominates $f$ on $\Omega=[0,1]^m$, namely $\check{f}(x) \geq f(x)$ for all $x \in \Omega$.

\noindent
{\em Proof:} Since $\Omega=[0,1]^m$ is the convex hull of $\L=\{0,1\}^m$, every $x \in \Omega$ is a convex combination
\[
x=\sum_{\lambda \in \L} t_\lambda \cdot \lambda.
\]
Since $\check{f}$ is affine and $f$ is convex, Claim 1 implies that 
\[
\check{f}(x) = 
\check{f}(\sum_{\lambda \in \L} t_\lambda \cdot \lambda) = 
\sum_{\lambda \in \L} t_\lambda \check{f}(\lambda) \geq
\sum_{\lambda \in \L} t_\lambda f(\lambda) \geq
f(\sum_{\lambda \in \L} t_\lambda \lambda) =
f(x).
\]
\hfill q.e.d 

Claim 2 implies  that for any $p \in \C M(\Omega,\B b)$
\[
E_p(f) \leq E_p(\check{f}) = E_{p}(\sum_{i=0}^m \alpha_i \ell_i) = \sum_{i=0}^m \alpha_i E_p(\ell_i) = \alpha_0+\sum_{i=1}^m \alpha_i b_i = E_{q^*}(f|_{\L}).
\]
However, $q^*$ extends to $\widetilde{q^*} \in \C M(\Omega,b)$ and clearly $E_{\widetilde{q^*}}(f)=E_{q^*}(f|_\L)$ so we get
\[
\max_{p \in \C M(\Omega,\B b)} E_p(f)  = E_{q^*}(f|_\L).
\]
Notice that $E_{q^*}(f|_\L) = \sum_{i=0}^m q^*(\rho_i)\cdot f(\rho_i)$.
Lemma \ref{L:approximation n=1} applied to $q^* \in \C M(\Omega,b)$ and the continuity of $f$ easily imply that there exist $\lambda \in \Mborel(\Omega,b)$ such that $E_\lambda(f)$ is arbitrarily close to $E_{q^*}(f)$.
If follows that $\sup_{\Mborel(\Omega,\B b)} E_p(f)=E_{q^*}(f|_{\L})$.
\end{proof}

\section{$(L,b)$-stable probability measures}
\label{SEC:Omega n}

\noindent
{\bf Main results.}
We consider $n$ iterations of the random vector $L=(\ell_1,\dots,\ell_m)$ over $\Omega=[0,1]^m$ from Section \ref{SEC:Omega}.
We obtain a sequence $L^1,\dots,L^n$ of random vectors in $\RR^m$ with sample space $\Omega^n$. 
In fact
\[
L^k \colon \Omega^n \to \Omega \subseteq \RR^m
\]
is the projection to the $k$th factor.
This is the universal process on $\Omega^n$, see Section \ref{SS:processes}. 
Denote by $\Mborel(\Omega,n)$ the set of all non-degenerate probability measures on $\Omega^n$.
We will identify these with the set of pdf's $p \colon \Omega^n \to \RR$ which are non-vanishing a.e.

With the notation and terminology of Section \ref{SS:conditional expectation} we make the following definition.

\begin{defn}
Consider some $b \in \OP{int}(\Omega)=(0,1)^m$.
A pdf $p \in \Mborel(\Omega,n)$ is called {\em $(L,b)$-stable} 
if for any $0 \leq k \leq n-1$ 
\[
E_p(L^{k+1}|L^1,\dots,L^k)=b.
\]
That is, the function $E_p(L^{k+1}|L^1,\dots,L^k) \colon \Omega^k \to \RR^m$ is constant with value $b$ a.e. 
The set of all $(L,b)$-stable $p \in \Mborel(\Omega,n)$ is denoted
\[
\Mborel(\Omega,n,b).
\]
\end{defn}

\begin{defn}
A function $f \colon \Omega^n \to \RR$ is called {\em fibrewise convex-supermodular} if it is convex-supermodular at each fibre.
Namely, for any $\tau \in \Omega^{k-1}$ and any $\theta \in \Omega^{n-k}$ the function $g \colon \Omega \to \RR$ defined by $g \colon \omega \mapsto f(\tau,\omega,\theta)$ is convex-supermodular (Definition \ref{D:convex-supermodular}).
\end{defn}

\begin{example}\label{Ex:fibrewise convex supermodular}
For any $J=(j_1,\dots,j_n) \in \C P_n(m)$, see \eqref{E:def Pkm}, let $\ell_J \colon \Omega^n \to \RR$ denote the function $\ell_J=\ell_{j_1} \otimes \cdots \otimes \ell_{j_n}$, namely
\[
\ell_J \colon (\omega^1,\dots,\omega^n) \mapsto \ell_{j_1}(\omega^1) \cdots \ell_{j_n}(\omega^n).
\]
Consider $f \colon \Omega^n \to \RR$ of the form $f=h \circ g$ where $h \colon \RR \to \RR$ is convex and $g \colon \Omega^n \to \RR$ is of the form
\[
g=\sum_{J \in \C P_n(m)} a_J \cdot \ell_J
\]
where $a_J \geq 0$ for all $J \neq (0,\dots,0)$.
Then $f$ is fibrewise convex-supermodular.

\noindent 
{\em Proof:}
It is clear that if $\tau \in \Omega^{k-1}$ and $\theta \in \Omega^{n-k}$ then the restriction of $\ell_J$ to the fibre $\{\tau\} \times \Omega \times \{\theta\} \subseteq \Omega^n$ is equal to $\sum_{i=0}^n b_i \ell_i$ where $b_i \geq 0$ for all $i \geq 1$; In fact $b_i$ is the sum of all $a_J$ in which the $k$th entry is equal to $i$.
The result follows from Example \ref{Ex:convex-supermodular}(\ref{Ex:convex-supermodular:1}).
\end{example}

The main results of this section are the following two theorems.

\begin{theorem}\label{T:infimum Jensen}
Let $f \colon \Omega^n \to \RR$ be a continuous function of the form $f=h \circ g$ in Example \ref{Ex:fibrewise convex supermodular}. 
Assume that $f \geq 0$.
Suppose that $b \in \OP{int}(\Omega)=(0,1)^m$.
Then
\[
\inf_{p \in \Mborel(\Omega,n,b)} E_p(f|L^1,\dots,L^k)(\omega^1,\dots,\omega^k)  = f(\omega^1,\dots,\omega^k,b,\cdots,b).
\]
\end{theorem}

Recall the upper supermodular vertex $q^*\colon \L \to \RR$ from Definition \ref{D:upper supermodular vertex}.
It extends to an atomic probability measure on $\Omega$ supported on $\L$ which we abusively denote $q^*$. 
Let
\[
q^*{}^{\otimes k}
\]
be the obvious product probability measure on $\L^k$ as well as its extension to $\Omega^k$.

Set $q_j = q^*(\rho_j)$ for all $0 \leq j \leq m$.
For any $J \in \C P_k(m)$ denote
\begin{eqnarray*}
q_J &=& \prod_{j \in J} q_j \\
\rho_J &=& (\rho_{j_1}, \cdots ,\rho_{j_k}) \in \Omega^k.
\end{eqnarray*}

If $f \colon \Omega^n \to \RR$ is measurable and $(\omega^1,\dots,\omega^k) \in \Omega^k$, we obtain a measurable function $g \colon \Omega^{n-k} \to \RR$ by $g(-)=f(\omega^1,\dots,\omega^k,-)$.
If $Q$ is a probability measure on $\Omega^{n-k}$ we will write $E_Q(f(\omega^1,\dots,\omega^k,-))$ for $E_Q(g)$.

\begin{theorem}\label{T:supremum Omega n}
Let $f \colon \Omega^n \to \RR$ be a continuous fibrewise convex-supermodular, $f \geq 0$.
Then
\begin{multline*}
\sup_{p \in \Mborel(\Omega,n,b)} E_p(f| L^1,\dots,L^k)(\omega^1,\dots,\omega^k) = E_{{q^*}^{\otimes n-k}}(f(\omega^1,\dots,\omega^k,-))
\\
=\sum_{J \in \C P_{n-k}(m)} q_J \cdot f(\omega^1,\dots,\omega^k,\rho_J).
\end{multline*}
\end{theorem}

\noindent 
In the remainder of this section we prove Theorems \ref{T:infimum Jensen} and \ref{T:supremum Omega n}.
Throughout we fix $b \in (0,1)^m$ and assume that it is non-increasing, namely $b_1 \geq \cdots \geq b_m$.

\begin{lemma}\label{L:Mborel n fact 1}
Consider some $p \in \Mborel(\Omega,n,b)$.
Suppose that $0 \leq k \leq n-1$ and consider $\omega^1,\dots,\omega^k \in \Omega$.
Set $p'=p_{L^{k+1}|(L^1,\dots,L^k)=(\omega^1,\dots,\omega^k)}$; See Section \ref{SS:conditional expectation}.
Then  $p' \in \Mborel(\Omega,b)$.
\end{lemma}

\begin{proof}
First, $p'$ is a pdf and $p'>0$ a.e., see \eqref{E:def pXI|XJ} in Section \ref{SS:conditional expectation}.
Set $I=\{1,\dots,k\}$ and $J=\{k+1\}$, subsets of $\{1,\dots,n\}$.
Write $L^I$ for the random vector $(L^1,\dots,L^k)$ and $\omega^I=(\omega^1,\dots,\omega^k) \in \Omega^k$.
We use Lemma \ref{L:conditional expectation in steps} and the fact that $p \in \Mborel(\Omega,n,b)$ to compute
\begin{multline*}
E_{p'}(L) = 
E_{p'}(\omega \mapsto L^{k+1}(\omega)) =
E_{p'}(\omega \mapsto E_p(L^{k+1}|L^{I \cup J})(\omega^I,\omega)) 
\\
=
E_p(L^{k+1}|L^I=\omega^I) =
E_p(L^{k+1}|L^1,\dots,L^k)(\omega^1,\dots,\omega^k) = b.
\end{multline*}
By definition, then, $p' \in \Mborel(\Omega,b)$.
\end{proof}

\begin{lemma}\label{L:tensors Lnb}
Let $p^1,\dots,p^n \in \Mborel(\Omega,b)$.
Then $p^1 \otimes \dots \otimes p^n \in \Mborel(\Omega,n,b)$.
\end{lemma}

\begin{proof}
It is clear that $p^1 \otimes \dots \otimes p^n$ is a non degenerate pdf on $\Omega^n$.
By Lemma \ref{L:tensors and conditional expectations}(\ref{L:tensors p:E XJ given XI}) and since by definition $E_{p^{k+1}}(L)=b$,
\[
E_{p^1 \otimes \dots \otimes p^n}(L^{k+1}|L^1,\dots,L^k)(\omega^1,\dots,\omega^k) = E_{p^{k+1}}(L^{k+1}) = E_{p^{k+1}}(L)=b.
\]
Since this holds for all $0 \leq k \leq n-1$, by definition $p^1 \otimes \dots \otimes p^n \in \Mborel(\Omega,n,b)$.
\end{proof}

\begin{lemma}[$n$-fold Approximation Lemma]
\label{L:n fold approximation lemma}
Let $b \in (0,1)^n$ and consider $q \in \C M(\Omega,b)$ with finite support $\{x^1,\dots,x^r\}$ and set $q_i=q(\{x^i\})$.
Let $q^{\otimes k}$ denote the induced product measure on $\Omega^k$.
Let $f \colon \Omega^n \to \RR$ be continuous with $f \geq 0$.
Then for any $\epsilon>0$ and any $0 \leq k \leq n$ there exists $P \in \Mborel(\Omega,n,b)$ such that 
\[
\Big| E_P(f|L^1,\dots,L^k)(\omega^1,\dots,\omega^k) - E_{q^{\otimes (n-k)}} (f(\omega^1,\dots,\omega^k,-))\Big| < \epsilon
\]
for all $\omega^1,\dots,\omega^k \in \Omega$.
\end{lemma}

\begin{proof}
Since $\Omega^n$ is compact, $f$ is bounded, i.e $\|f\|_\infty<\infty$.
Since $f$ is uniformly continuous, we choose $\delta>0$ suitable for $\tfrac{\epsilon}{3n}$.
Apply Lemma \ref{L:approximation n=1} with $\delta$ and with $\tfrac{\epsilon}{3n\|f\|_\infty}$ to obtain $p \in \Mborel(\Omega,b)$ and $\beta<\tfrac{\epsilon}{3n\|f\|_\infty}$ such that for any continuous $g \colon \Omega \to \RR$ where $g \geq 0$,
\begin{equation}\label{E:n-fold approximation:p}
E_p(g) = \beta \int_\Omega g \, d\mu + \sum_{i=1}^r (q_i-\tfrac{\beta}{r}) g(\xi^i)
\end{equation}
for some $\xi^1,\dots,\xi^r \in \Omega$ such that $\|x^i-\xi^i\|_\infty<\delta$.

Set $[r]=\{1,\dots,r\}$.
For any $I=(i_1,\dots,i_{n-k}) \in [r]^{n-k}$ set 
\[
q_I=q_{i_1} \cdots q_{i_{n-k}} \qquad \text{and} \qquad x^I=(x^{i_1},\dots,x^{i_{n-k}}) \in \Omega^{n-k}
\] 
and let $f_I \colon \Omega^{k} \to \RR$ be the function 
\[
f_I \colon (\omega^1,\dots,\omega^{k}) \mapsto f(\omega^1,\dots,\omega^{k},x^I).
\]
Observe that 
\begin{equation}\label{E:n-fold apprx:Eq tensor}
E_{q^{\otimes(n-k)}}(f(\omega^1,\dots,\omega^k,-)) = \sum_{I \in [r]^{n-k}} q_I \cdot f_I(\omega^1,\dots,\omega^k).
\end{equation}
It is clear that $\|f_I\|_\infty \leq \|f\|_\infty $ and that $f_I$ is uniformly continuous with the same $\delta$ suitable for $\tfrac{\epsilon}{3n}$ as that for $f$.
Recall $p \in \Mborel(\Omega,b)$ that we chose at the start of the proof.

\noindent
{\em Claim:} Consider some $0 \leq k <n$ and some $I \in [r]^{n-k-1}$.
Then for any $\omega^1,\dots,\omega^k \in \Omega$
\[
\left| E_p\left(\omega \mapsto f_I(\omega^1,\dots,\omega^k,\omega)\right) - \sum_{i=1}^r q_i \cdot f_I(\omega^1,\dots,\omega^k,x^i)\right|<\tfrac{\epsilon}{n}.
\]
\noindent
{\em Proof:} Set $g(w)=f_I(\omega^1,\dots,\omega^k,\omega)$.
Clearly $g \colon \Omega \to \RR$ is continuous and $\|g\|_\infty \leq \|f\|_\infty$.
Moreover, it is uniformly continuous and clearly the same $\delta$ we chose for $f$ suitable for $\tfrac{\epsilon}{3n}$ works for $g$.
Since $E_{q}(g) = \sum_{i=1}^r q_i g(x^i)$ and since \eqref{E:n-fold approximation:p} holds 
\[
|E_p(g)-E_q(g)| \leq
\beta \int_\Omega g\, d\mu +\tfrac{\beta}{r}\sum_{i=1}^r |g(\xi^i)| +\sum_{i=1}^r q_i |g(\xi^i)-g(x^i)|.
\] 
Since $\|g\|_\infty\leq \|f\|_\infty$ and $\beta<\tfrac{\epsilon}{3n\|f\|_\infty}$ the first and second terms in this sum are less than $\tfrac{\epsilon}{3n}$.
Since $\|\xi^i-x^i\|_\infty<\delta$, the uniform continuity of $g$ implies that the same is true for the last term since $\sum_i q_i=1$.
This completes the proof of the claim.
\hfill q.e.d

Set $P=p^{\otimes n}$.
Then $P \in \Mborel(\Omega,n,b)$ by Lemma \ref{L:tensors Lnb}.
In light of \eqref{E:n-fold apprx:Eq tensor}, we complete the proof of the lemma by showing by downward induction on $0 \leq k \leq n$ that
\begin{equation}\label{L:n fold approximation:induction}
\left| E_{p^{\otimes n}}(f|L^1,\dots,L^k)(\omega^1,\dots,\omega^k) - \sum_{I \in [r]^{n-k}} q_I \cdot f_I(\omega^1,\dots,\omega^k)\right| \leq (n-k)\tfrac{\epsilon}{n}.
\end{equation}
The base of induction $k=n$ is a triviality since
\[
E_{p^{\otimes n}}(f|L^1,\dots,L^n)(\omega^1,\dots,\omega^n) = f(\omega^1,\dots,\omega^n)
\]
and since $f_I=f$ and $q_I=1$ for the the only $I \in [r]^0$. 

Assume inductively that \eqref{L:n fold approximation:induction} holds for some $1 \leq k \leq n$.
Fix some $\omega^1,\dots,\omega^{k-1} \in \Omega$.
Lemmas \ref{L:conditional expectation in steps} and \ref{L:tensors and conditional expectations}(\ref{L:tensors p:q XI given XJ})  imply that for any 
\[
E_p(\tau \mapsto E_{p^{\otimes n}}(f|L^1,\dots,L^{k})(\omega^1,\dots,\omega^{k-1},\tau)) =
E_{p^{\otimes n}}(f|L^1,\dots,L^{k-1})(\omega^1,\dots,\omega^{k-1}).
\]
Viewing the left hand side of \eqref{L:n fold approximation:induction} with $\omega^1, \dots,\omega^{k-1}$ fixed as a function of $\tau \in \Omega$, the linearity of expectation $E_p(-)$ implies 
\[
\left|E_{p^{\otimes n}}(f|L^1,\dots,L^{k-1})(\omega^1,\dots \omega^{k-1}) - \sum_{I \in [r]^{n-k}} q_I E_p(f_I(\omega^1,\dots,\omega^{k-1},-))\right|< (n-k)\tfrac{\epsilon}{n}.
\]
Thanks to \eqref{E:n-fold apprx:Eq tensor}, in order to complete the induction step (to $k-1$) it remains to show by  that
\[
\left| \sum_{I \in [r]^{n-k}} q_I \cdot E_p(f_I(\omega^1,\dots,\omega^{k-1},-)) - \sum_{J \in [r]^{n-k+1}} q_J \cdot  f_J(\omega^1,\dots,\omega^{k-1}) \right|< \tfrac{\epsilon}{n}.
\]
Given $J=(i_1,\dots,i_{n-k+1}) \in [r]^{n-k+1}$ set $I=(i_2,\dots,i_{n-k+1}) \in [r]^{n-k}$ and observe that $q_J=q_I q_{i_1}$ and that $f_J(\omega^1,\dots,\omega^{k-1})=f_I(\omega^1,\dots,\omega^{k-1},x^{i_1})$.
By the Claim above
\begin{align*}
\Big| \sum_{I \in [r]^{n-k}} & q_I E_p(f_I(\omega^1,\dots,\omega^{k-1},-)) - \sum_{J \in [r]^{n-k+1}} q_J f_J(\omega^1,\dots,\omega^{k-1}) \Big| = 
\\
& =
\Big| \sum_{I \in [r]^{n-k}} q_I E_p(f_I(\omega^1,\dots,\omega^{k-1},-)) - \sum_{I \in [r]^{n-k}} \sum_{i=1}^r q_I q_i f_I(\omega^1,\dots,\omega^{k-1},x^i) \Big| 
\\
& \leq
\sum_{I \in [r]^{n-k}} q_I \cdot \Big| E_p(f_I(\omega^1,\dots,\omega^{k-1},-)) - \sum_{i=1}^r q_i f_I(\omega^1,\dots,\omega^{k-1},x^i) \Big| 
\\
& <
\sum_{I \in [r]^{n-k-1}} q_I \cdot \tfrac{\epsilon}{n} = \tfrac{\epsilon}{n}.
\end{align*}
This completes the induction step.
\end{proof}

\begin{proof}[Proof of Theorem \ref{T:infimum Jensen}]
For any $J=(j_1,\dots,j_n) \in \C P_n(m)$ and any $\omega^1, \dots, \widehat{\omega^k},\dots, \omega^n \in \Omega$ (meaning $\omega^k$ is omitted) we have
\[
\ell_J(\omega^1,\dots,\omega^{k-1},-,\omega^{k+1},\dots,\omega^n) = \prod_{i \neq k} \ell_{j_i}(\omega^i) \cdot \ell_{j_k}(-).
\]
Therefore, if $\lambda \in \Mborel(\Omega,b)$ we get
\begin{multline*}
E_\lambda( \ell_J(\omega^1,\dots,\omega^{k-1},-,\omega^{k+1},\dots,\omega^n)) = 
 E_\lambda(\ell_{j_k}) \cdot \prod_{i \neq k} \ell_{j_i}(\omega^i) =
b_{j_k} \cdot \prod_{i \neq k} \ell_{j_i}(\omega^i) =
\\
\ell_{j_k}(b) \cdot \prod_{i \neq k} \ell_{j_i}(\omega^i) =
\ell_J(\omega^1,\dots,\omega^{k-1},b,\omega^{k+1},\dots,\omega^n).
\end{multline*}

We use downward induction on $0 \leq k \leq n$ to show that for any $p \in \Mborel(\Omega,n,b)$
\[
E_p(f|L^1,\dots,L^k)(\omega^1,\dots,\omega^k) \geq f(\omega^1,\dots,\omega^k,b,\dots,b)
\]
almost everywhere.
The base of induction $k=n$ is a triviality (and in fact, equality holds a.e).
Assume the inequality holds for $k+1 \leq n$. 
Set $p'=p_{L^{k+1}|L^1=\omega^1,\dots,L^k=\omega^k}$.
Then $p' \in \Mborel(\Omega,b)$ by Lemma \ref{L:Mborel n fact 1}.
Lemma \ref{L:conditional expectation in steps} and the induction hypothesis imply
\begin{align*}
E_p(f|L^1,\dots,L^k)(\omega^1,\dots,\omega^k) & =
E_{p'}(\omega \mapsto E_p(f|L^1,\dots,L^{k+1})(\omega^1,\dots,\omega^k,\omega)) 
\\
& \geq 
E_{p'}(\omega \mapsto f(\omega^1,\dots,\omega^k,\omega,b,\dots,b).
\end{align*}
Since $f=h\circ g$ with $h$ convex and $g$ as in Example \ref{Ex:fibrewise convex supermodular}, Jensen's inequality allows us to continue the inequality
\begin{align*}
& \geq 
h(\sum_{J} a_J E_{p'}(\ell_J(\omega^1,\dots,\omega^{k},-,b,\dots,b)) 
\\
& =
h(\sum_J a_J \ell_J(\omega^1,\dots,\omega^k,b,\dots,b)) 
\\
& =
h(g(\omega^1,\dots,\omega^k,b,\dots,b)) 
\\
& = 
f(\omega^1,\dots,\omega^k,b,\dots,b).
\end{align*}
This completes the induction step.

We deduce that in the statement of the theorem the right hand side is a lower bound for the left hand side and it remains to show equality.
Let $\nu$ be the probability measure on $\Omega$ supported on $\{b\}$, i.e $\nu(\{b\})=1$.
It is clear that for any measurable function $g \colon \Omega^k \to \RR$ we have $E_{\nu^{\otimes k}}(g)=g(b,\dots,b)$.
By Lemma \ref{L:n fold approximation lemma}, for any $\epsilon>0$ there exists $P \in \Mborel(\Omega,n,b)$ such that $E_P(f|L^1,\dots,L^k)(\omega^1,\dots,\omega^k)$ is $\epsilon$-close to $E_{\nu^{\otimes (n-k)}}(f(\omega^1,\dots,\omega^k,-))=f(\omega^1,\dots,\omega^k,b,\dots,b)$.
This completes the proof.
\end{proof}

\begin{proof}[Proof of Theorem \ref{T:supremum Omega n}]
First, observe that 
\begin{equation}\label{E:supremum Omega n:Eq*tensor}
E_{{q^*}^{\otimes (n-k)}}(f(\omega^1,\dots,\omega^k,-)) = \sum_{I \in \C P_{n-k}(m)} q_I \cdot f(\omega^1,\dots,\omega^k,\rho_I).
\end{equation}
Next, we prove that for any $p \in \Mborel(\Omega,n,b)$ and any $0 \leq k \leq n$
\begin{equation}\label{E:supremum:induction}
E_p(f| L^1,\dots,L^k)(\omega^1,\dots,\omega^k)  \leq E_{(q^*)^{\otimes n-k}}(f(\omega^1,\dots,\omega^k,-)).
\end{equation}
Fix some $p$ and use downward induction on $k$.
The base of induction $k=n$ is a triviality since $E_p(f|L^1,\dots,L^n)=f$ a.e.
Assume inductively that \eqref{E:supremum:induction} holds for $k+1 \leq n$.
Set $p'=p_{L^{k+1}|(L^1,\dots,L^k)=(\omega^1,\dots,\omega^k)}$.
Lemma \ref{L:conditional expectation in steps} and the induction hypothesis together with \eqref{E:supremum Omega n:Eq*tensor} imply that
\begin{align*}
E_p(f|L^1,\dots,L^k)(\omega^1,\dots,\omega^k) &= 
E_{p'}(\omega \mapsto E_p(f|L^1,\dots,L^{k+1})(\omega^1,\dots,\omega^k,\omega)) 
\\
& \leq
E_{p'}(\omega \mapsto \sum_{I \in \C P_{n-k-1}}  q_I f(\omega^1,\dots,\omega^k,\omega,\rho_I)) 
\\
& =
\sum_{I \in \C P_{n-k-1}(m)} q_I \cdot E_{p'} (\omega \mapsto f(\omega^1,\dots,\omega^k,\omega,\rho_I)).
\end{align*}
By the assumption on $f$, each function $f(\omega^1,\dots,\omega^k,-,\rho_I)$ is convex-supermodular and continuous and.
Lemma \ref{L:Mborel n fact 1} and Theorem \ref{T:one step max} allow us to continue the estimate
\[ 
\leq 
\sum_{I \in \C P_{n-k-1}(m)} q_I \cdot  \sum_{j=0}^m q_j \cdot f(\omega^1,\dots,\omega^k,\rho_j,\rho_I)) =
\sum_{I \in \C P_{n-k}(m)} q_I \cdot f(\omega^1,\dots,\omega^k,\rho_I).
\] 
Together with \eqref{E:supremum Omega n:Eq*tensor}, this completes the induction step.

We deduce that for any $0 \leq k \leq n$ the right hand side in the statement of the theorem is an upper bound for the left hand side. 
By Lemma \ref{L:n fold approximation lemma} there exist $P \in \Mborel(\Omega,n,b)$ such that $E_P(f|L^1,\dots,L^k)(\omega^1,\dots,\omega^k)$ are arbitrarily close to $E_{{q^*}^{\otimes (n-k)}}(f(\omega^1,\dots,\omega^k,-)$.
This completes the proof.
\end{proof}

\section{Proofs of the main results}

In this section we prove the results in Section \ref{Sec:main results}.
We start by setting up a formal framework for the discrete-time continuous-binomial market model presented there.

We begin with the ``one-step'' process, namely description of the price jumps $\Psi_i$ where $0 \leq i \leq m$.
By definition $\Psi_0=R$ and $\Psi_i$ are chosen at random from the interval $[D_i,U_i]$.
By choosing a linear homeomorphisms $[D_i,U_i] \cong [0,1]$, a natural sample space for the probability space underlying a single step is $\Omega=[0,1]^m$ and
\begin{eqnarray*}
&& \Psi_i(x_1,\dots,x_m) = D_i + (U_i-D_i)x_i, \\
&& \Psi_0(x_1,\dots,x_m) = R
\end{eqnarray*}
With the notation of Section \ref{SEC:Omega}, for any $1 \leq i \leq m$
\[
\Psi_i=D_i \ell_0 + (U_i-D_i) \ell_i.
\]
We will write $\Psi \colon \Omega \to \RR^{m+1}$ for the random vector 
\[
\Psi = (\Psi_0,\dots,\Psi_m).
\]

The natural sample space for the $n$-step model is $\Omega^n$.
We obtain a process $\Psi^1,\dots,\Psi^n$ of the price changes at time $k$:
\[
\Psi^k \colon \Omega^n \xto{L^k} \Omega \xto{\Psi} \RR^{m+1}
\]
where $L^k$ is the projection to the $k$-th factor and $L^1,\dots,L^k$ form the universal process on $\Omega^n$, see Section \ref{SS:processes}.
Thus,
\[
\Psi_i^k=D_i+(U_i-D_i)L_i^k
\]
where $L^k_i$ is the $i$th component of $L^k \colon \Omega^n \to \Omega \subseteq \RR^{m+1}$ and we observe that (since $\ell_0=\mathbf{1} \colon \Omega \to \RR$)
\[
L^k_i = \ell_0^{\otimes (k-1)} \otimes \ell_i \otimes \ell_0^{\otimes (n-k-1)}.
\]
Recall that we assume that $0<D_i<R<U_i$ so in particular $\Psi_i^k>0$ for all $i$ and all $k$.

The prices of the assets form an  $\RR^{m+1}$-valued process
\[
S^0,\dots,S^n \colon \Omega^n \to \RR^{m+1}
\]
where $S^k=(S_0^k,\dots,S_m^k)$ is the vector of prices of the assets at time $k$.
It is assumed by the model that 
\[
S_i^k >0 \qquad \text{for all $0 \leq i \leq m$ and $0 \leq k \leq n$.}
\]
By construction of the model, the processes $S^0,\dots,S^n$ and $\Psi^1,\dots,\Psi^n$ satisfy the relation
\[
S_i^k = S_i^{k-1} \cdot \Psi_i^k \qquad \text{($0 \leq i \leq m$ and $1 \leq k \leq n$).}
\]
It is therefore clear that for any $0 \leq k \leq n$
\[
S_i^k = S_i^0 \cdot \Psi_i^1 \cdots \Psi_i^k.
\]
Recall the  definition of $\C P_k(m)$ from \eqref{E:def Pkm} in Section \ref{Sec:main results}.
For $J=(j_1,\dots,j_k) \in \C P_k(m)$ set $\ell_J=\ell_{j_1} \otimes \cdots \otimes \ell_{j_k}$.
It follows that
\[
S_i^k = \sum_{J \in \C P_k(m)} a_J \cdot L^1_{j_1} \cdots L^k_{j_k} =
\sum_{J \in \C P_k(m)} a_J \cdot \ell_J \otimes \underbrace{\mathbf{1} \otimes \cdots \otimes \mathbf{1}}_{\text{$n-k$ times}}
\]
for some $a_J \geq 0$. 

\noindent
{\bf Comment:} In Section \ref{Sec:main results} the processes $\Psi^k$ and $S^k$ were denoted $\Psi(k)$ and $S(k)$.

The European basket in Section \ref{Sec:main results} is the function (random variable) $F \colon \Omega^n \to \RR$
\[
F=\big(\underbrace{\sum_{i=0}^m c_i \cdot S_i^n -K}_G\big)^+
\]
where $c_i \geq 0$ for $1 \leq i \leq m$ and $K>0$ is some number.
Notice that 
\[
G= \sum_{J \in \C P_n(m)} a_J \ell_J
\]
where $a_J \geq 0$ for all $J \neq (0,\dots,0)$.
It follows from Example \ref{Ex:fibrewise convex supermodular} and since $h(x)=x^+$ is continuous, convex and non-negative that
\begin{prop}\label{P:European basket is fibrewise convex-supermodular}
$F$ is continuous and fibrewise convex-supermodular and $F \geq 0$.
\end{prop}

Recall that a non-degenerate probability measure $p$ on $\Omega^n$ is called {\em risk neutral} if for any $k \geq 0$ and any $j \geq 1$ such that $k+j \leq n$
\[
E_p(S^{k+j}|L^1,\dots,L^k)(\omega^1,\dots,\omega^k) = R^j \cdot S^k(\omega^1,\dots,\omega^k).
\]
We denote the set of these probability measures by $\OP{RN}$.

\begin{prop}\label{P:RN=Mborel(Omega,n,b)}
$\OP{RN}=\Mborel(\Omega,n,b)$ where $b=(b_1,\dots,b_m)$ is defined in \eqref{Ebi from DU} in Section \ref{Sec:main results}.
\end{prop}

\begin{proof}
Since $S_i^{k+j}=S_i^k \cdot \Psi_i^{k+1} \cdots \Psi_i^{k+j}$ and since $S_i^k>0$, it is clear that the condition for $p$ being a risk neutral measure is equivalent to the condition
\[
E_p(\Psi_i^{k+1}\cdots\Psi_i^{k+j}|L^1,\dots,L^k) = R^j
\]
(almost everywhere constant function $\Omega^{n-k} \to \RR$).
It is easily verified using induction and Lemmas \ref{L:conditional expectation of truncated functions} and \ref{L:conditional expectation in steps} that this condition (for any $k,j$ such that $k+j \leq n$) is equivalent to the single step condition, namely
\[
E_p(\Psi_i^{k+1}|L^1,\dots,L^k) = R
\]
for all $1 \leq k \leq n-1$.
But $\Psi_i^{k+1} = D_i + (U_i-D_i) \cdot L_i^{k+1}$.
So the condition above is equivalent to 
\[
D_i  + (U_i-D_i) \cdot E_p(L_i^{k+1} | L^1,\dots,L^k) =R.
\]
Using the definition of $b_1,\dots,b_m$ in \eqref{Ebi from DU}, this is equivalent to  $E_p(L_i^{k+1} | L^1,\dots,L^k) = b_i$, and collecting these for all $1 \leq i \leq m$ we get 
\[
E_p(L^{k+1} | L^1,\dots,L^k)=b
\]
which by definition is the condition for $p \in \Mborel(\Omega,n,b)$.
\end{proof}

\begin{proof}[Proof of Theorem \ref{T:rational values interval ends}]
By Proposition \ref{P:European basket is fibrewise convex-supermodular} $F$ is continuous convex-supermodular and $F \geq 0$.
The interval $(\Gamma_{\min}(F,k)\, , \, \Gamma_{\max}(F,k))$ of the rational values of $F$ at some state of the world $(\omega^1,\dots,\omega^k) \in \Omega^k$ is known to be the collection of numbers 
\[
\{ \ R^{k-n} \cdot E_p(F|L^1,\dots,L^k)(\omega^1,\dots,\omega^k) \ \}_{p \in \OP{RN}}.
\] 
Proposition \ref{P:RN=Mborel(Omega,n,b)} and Theorem \ref{T:infimum Jensen} imply that
\begin{multline*}
R^{n-k} \cdot \Gamma_{\min}(F,k)(\omega^1,\dots,\omega^k) = 
\inf_{p \in \OP{RN}} E_p(F|L^1,\dots,L^k)(\omega^1,\dots,\omega^k) =
\\
\inf_{p \in \Mborel(\Omega,n,b)} E_p(F|L^1,\dots,L^k)(\omega^1,\dots,\omega^k) =
F(\omega^1,\dots,\omega^k,b,\dots,b).
\end{multline*}
By definition of $b$, see \eqref{Ebi from DU}, and since $\Psi_i=D_i\ell_0+(U_i-D_i)\ell_i$,
\[
\Psi_i(b)=R \qquad \text{ for all $1 \leq i \leq m$.}
\]
By definition, for any $1 \leq i \leq m$
\[
S_i^n(\omega^1,\dots,\omega^k,b,\dots,b) = 
S_i^0 \cdot \Psi_i(\omega^1) \cdots \psi_i(\omega^k) \cdot \underbrace{\Psi_i(b) \cdots \Psi_i(b)}_{\text{$n-k$ times}} =
R^{n-k} S_i^k(\omega^1,\dots,\omega^k).
\]
Also, for $i=0$ we clearly get $S_0^n=S_0^0 \cdot R^n=R^{n-k} S_0^k$.
It follows that 
\[
F(\omega^1,\dots,\omega^k,b,\dots,b) =
\left( R^{n-k} \sum_{i=0}^m c_i S_i^k -K\right)^+(\omega^1,\dots,\omega^k).
\]
This establishes the formula for $\Gamma_{\min}(F,k)$.

We note that the numbers $q_i$ defined in \eqref{E:def qi for main theorem} and used in the statement of the theorem are equal to $q^*(\rho_i)$ of the upper supermodular vertex (Definition \ref{D:upper supermodular vertex}).
Since by Proposition \ref{P:European basket is fibrewise convex-supermodular} the conditions of Theorem \ref{T:supremum Omega n} hold, it follows that 
\begin{align*}
R^{n-k} \cdot\Gamma_{\max}(F,k)(\omega^1,\dots,\omega^k) & = 
\sup_{p \in \OP{RN}} E_p(F|L^1,\dots,L^k)(\omega^1,\dots,\omega^k) 
\\
& =
\sup_{p \in \Mborel(\Omega,n,b)} E_p(F|L^1,\dots,L^k)(\omega^1,\dots,\omega^k) 
\\
& =
\sum_{J \in \C P_{n-k}(m)} q_J \cdot F(\omega^1,\dots,\omega^k,\rho_{j_1}, \dots, \rho_{j_{n-k}}).
\end{align*}
Since $\ell_i(\rho_j)=1$ if $i \leq j$ and $\ell_i(\rho_j)=0$ if $i>j$ it follows that $\Psi_i(\rho_j)=D_i+(U_i-D_i)\ell_i(\rho_j)=\chi_i(j)$.
Therefore, for any $J \in \C P_{n-k}(m)$,
\begin{multline*}
S_i^n(\omega^1,\dots,\omega^k,\rho_{j_1}, \dots, \rho_{j_{n-k}}) =
S_i^0 \cdot \Psi_i(\omega^1) \cdots \Psi_i(\omega^k) \cdot \Psi_i(\rho_{j_1}) \cdots \Psi_i(\rho_{j_{n-k}}) =
\\
\chi_i(J) \cdot S_i^k(\omega^1,\dots,\omega^k).
\end{multline*}
For $i=0$ we get, of course, $S_0^n=S_0^0 \cdot R^n=S_0^k \cdot \chi_0(J)$.
Substitution into the definition of $F$ we get
\[
F(\omega^1,\dots,\omega^k,\rho_{j_1}, \dots, \rho_{j_{n-k}}) =
\left(\sum_{i=0}^m c_i \cdot \chi_i(J) \cdot S_i^k -K\right)^+ (\omega^1,\dots,\omega^k).
\]
This establishes the formula for $\Gamma_{\max}(F,k)$.
\end{proof}

\begin{proof}[Proof of Theorem \ref{T:shrinking ends}]
Recall that $\Psi_i(b)=D_i\ell_0-(U_i-D_i)\ell_i(b)=D_i+(U_i-D_i)b_i=R$ for $1 \leq i \leq m$ and that $\Psi_0=R$ by definition.
By Theorem \ref{T:rational values interval ends}
\[
\Gamma_{\min}(F,0) = 
R^{-n} \cdot F(b,\dots,b) =
R^{-n}(\sum_{i=0}^m c_i S_i^0 \cdot \Psi_i(b)^n - K)^+ = 
R^{-n} \cdot \left(R^n \sum_{i=0}^m c_i S_i^0 - K\right)^+.
\]
This is independent of $U_i, D_i$.

Set $b_i$ as in \eqref{Ebi from DU} and denote by $b_i(s)$ the values of $b_i$ in our market model with parameters $u_i(s)$ and $d_i(s)$.
Notice that $d_i(s)<R$ and that $u_i(s)>R$ for all $0 \leq s <1$.
Moreover, $d_i(0)=D_i$ and $u_i(0)=U_i$ and $\lim_{s \nearrow 1} d_i(s) =R$ and $\lim_{s \nearrow 1} u_i(s) = R$.
One checks that 
\[
b_i(s)=\frac{R-d_i(s)}{u_i(s)-d_i(s)} = b_i.
\]
Therefore the values of $q_i=b_i-b_{i+1}$ are independent of $s$.
The values of $\chi_i(j)$ do depend on $s$ where $\chi_i(j)(s)=u_i(s)$ if $i \leq j$ and $\chi_i(j)(s)=d_i(s)$ if $i> j$.
Thus, $\chi_i(j)(s)$ is a polynomial (of degree $1$) in $s$.
Moreover, $\lim_{s \nearrow} \chi_i(j)(s)=R$.
By Theorem \ref{T:rational values interval ends}
\[
\varphi(s) = 
\Gamma_{\max}(F,0;u_i(s),d_i(s)) = 
R^{-n}\sum_{J \in \C P_n(m)} q_J \cdot \left(\sum_{i=0}^m c_i \cdot S_i^0 \cdot \prod_{j \in J} \chi_i(j)(s)\right)^+.
\]
So $\vp$ is a continuous function of $s \in [0,1)$.
Now, $\vp(0)=\Gamma_{\max}(F,0;U_i,D_i)$ since $d_i(0)=D_i$ and $u_i(0)=U_i$.
Since $h \colon x \mapsto x^+$ is continuous and $\sum_{j=0}^m q_j=1$ we get
\begin{align*}
\lim_{s \nearrow 1} \varphi(s) &= 
R^{-n} \sum_{J \in \C P_n(m)} q_J \left(\sum_{i=0}^m c_i S_i^0\cdot \lim_{s \nearrow 1}\prod_{j \in J} \chi_i(j)(s)-K\right)^+ 
\\
& =
R^{-n}\sum_{J \in \C P_n(m)} q_J \left(\sum_{i=0}^m c_i S_i^0 \cdot R^n-K\right)^+ 
\\
& =
R^{-n}\left(\sum_{i=0}^m c_i S_i^0 \cdot R^n-K\right)^+ 
\\
&=
\Gamma_{\min}(F,0;U_i,D_i).
\end{align*}
The ``intermediate value'' result in the theorem  follows from the continuity of $\varphi(s)$.
\end{proof}

In the next lemma we will consider processes on $\Omega$, see Section \ref{SS:processes}.
Recall that the universal process is $L^1,\dots,L^n$ where $L^k$ is the projection $\Omega^n \to \Omega$ to the $k$-th factor followed by the inclusion to $\RR^m$.

\begin{lemma}\label{L:linear programming for hedging}
Fix some $n$.
Let $q$ be a probability measure on $\Omega$ supported on $\{\rho_0,\dots,\rho_m\}$, see \eqref{E:define rho_k}.
Consider $\RR$-valued processes
\[
r^0,\dots,r^{n-1} >0 \qquad \text{and} \qquad X^0,\dots,X^n
\]
and $\RR^m$-valued processes
\[
d^0,\dots,d^{n-1} \qquad \text{and} \qquad
\Delta^0,\dots,\Delta^{n-1} > 0.
\]
Further, consider an $\RR^{m+1}$-valued processes
\[
S^0,\dots,S^n \qquad \text{and} \qquad \Psi^1,\dots,\Psi^n.
\]
with components $S^k=(S_0^k,\dots,S_m^k)$ and $\Psi^k=(\Psi_0^k,\dots,\Psi_m^k)$.
Assume that
\begin{enumerate}
\item
$S_i^k=S_i^{k-1} \cdot \Psi_i^k$ for all $1 \leq k \leq n$ and all $0 \leq i \leq m$.

\item
$\Psi_0^{k+1}=r^k$ and $\Psi_i^{k+1}=d_i^k+\Delta_i^k \cdot L_i^{k+1}$ for any $0 \leq k \leq n-1$, and for any $\omega^K=(\omega^1,\dots,\omega^k) \in \Omega^k$,
\[
E_q(\omega \mapsto \Psi_i^{k+1}(\omega^K,\omega)) = r^k(\omega^K).
\]

\item 
$X^k$ are fibrewise convex-supermodular, and for any $0 \leq k \leq n-1$ and any $\omega^K \in \Omega^{k}$
\[
E_q(\omega \mapsto X^{k+1}(\omega^K,\omega) ) = r^k(\omega^K) \cdot X^k(\omega^K).
\]
\end{enumerate}
Then among all $\RR^{m+1}$-valued processes $\beta^0,\dots,\beta^{n-1}$ there exists a process $\alpha^0,\dots,\alpha^{n-1}$ which minimises the $\RR$-valued random process
\begin{equation}\label{E:L:linear programming:V}
V^k_\beta = \sum_{i=0}^m \beta^k_i \cdot S_i^k
\end{equation}
subject to the condition 
\begin{equation}\label{E:L:linear programming:constraint}
\sum_{i=0}^m \beta^k_i \cdot S_i^{k+1} \geq X^{k+1}.
\end{equation}
Moreover, $V_\alpha^k=X^k$ and the value of $\alpha^k(\omega^K)$ at $\omega^K \in \Omega^k$ can be computed by
\[
\begin{bmatrix}
\alpha^k_0(\omega^K) \\
\alpha^k_1(\omega^K) \\
\vdots \\
\alpha^k_m(\omega^K)
\end{bmatrix}
=
T(\omega^K)^{-1} \cdot M'(\omega^K)^{-1} \cdot Q \cdot 
\begin{bmatrix}
X^{k+1}(\omega^K \rho_0) \\
X^{k+1}(\omega^K \rho_1) \\
\vdots \\
X^{k+1}(\omega^K \rho_m)
\end{bmatrix}
\]
where $Q$ is the matrix in the statement of Theorem \ref{T:maximal hedging} and $T,M' \colon \Omega^k \to \OP{Mat}_{(m+1)\times (m+1)}(\RR)$ are the $(m+1) \times (m+1)$ matrices
\begin{eqnarray*}
&& T = 
\begin{bmatrix}
r^{k} \cdot S_0^k  \\
      & S_1^k \\
      &      &  \ddots \\
      &      &         & S_m^k 
\end{bmatrix}
\\
&&
M'=\left[
\begin{array}{c|cccc}
1 & d_1^k        & d_2^k    & \cdots & d_m^k      \\
\hline
0 & \Delta_1^k                                   \\
\vdots          &          & \ddots              \\
0 &             &          &        & \Delta_m^k
\end{array}\right]
\end{eqnarray*}
\end{lemma}

\begin{proof}
We prove the lemma in a sequence of claims.

\noindent
{\em Claim 1:} If a process $\beta^0,\dots,\beta^{n-1}$ satisfies \eqref{E:L:linear programming:constraint} then $V_\beta^k \geq X^k$ for all $0 \leq k \leq n-1$.

\noindent
{\em Proof:}
Choose some $k$ and some $\omega^K \in \Omega^k$.
Then \eqref{E:L:linear programming:constraint} becomes the following system of (infinitely many) inequalities in the unknowns $\beta_i^k(\omega^K)$
\begin{equation}\label{E:L:linear programming:equivalent-constraints}
\sum_{i=0}^m \beta^k_i(\omega^K) \cdot S_i^{k+1}(\omega^K, \omega) \geq X^{k+1}(\omega^K, \omega), \qquad (\omega \in \Omega).
\end{equation}
Let $\Phi(\omega)$ denote the left hand side of \eqref{E:L:linear programming:equivalent-constraints} and $Y(\omega)$ denote its right hand side.
These are random variables with domain $\Omega$ so 
$E_q(\Phi) \geq E_q(Y)$.
By the hypotheses on $X^k$
\[
E_q(Y) = E_q(\omega \mapsto X^{k+1}(\omega^K \omega) = r^{k}(\omega^K) \cdot X^k(\omega^K).
\]
By the hypotheses on $S^k$ and $\Psi^k$
\begin{multline*}
E_q(\Phi) = 
\sum_{i=0}^m \beta_i^k(\omega^K) \cdot S_i^k(\omega^K) \cdot E_q(\omega \mapsto \Psi^{k+1}_i(\omega^K,\omega)) = 
\\
r^{k}(\omega^K) \cdot \sum_{i=0}^m \beta_i^k(\omega^K) \cdot S_i^k(\omega^K) = 
r^{k}(\omega^K) \cdot V_\beta^k(\omega^K).
\end{multline*}
Since $r^{k}>0$ it follows that $V_\beta^k(\omega^K) \geq X^k(\omega^K)$.
\hfill q.e.d

Consider some $0 \leq k \leq n-1$ and some $\omega^K=(\omega^,\dots,\omega^k) \in \Omega^k$.
We will show in Claim 5 below that the system of $m+1$ linear equations with $m+1$ unknowns $\alpha^k_0(\omega^K),\dots,\alpha^k_m(\omega^K)$
\begin{equation}\label{E:L:linear programming:vertex}
\sum_{i=0}^m S_i^k(\omega^K) \Psi^{k+1}_i(\omega^K,\rho_j) \cdot \alpha^k_i(\omega^K) = X^{k+1}(\omega^K,\rho_j), \qquad 0 \leq j \leq m
\end{equation}
has a unique solution given by the matrices $Q,M',T$ as in the statement of the lemma.
Since $S^k>0$ and $d^k$ and $\Delta^k>0$ are measurable, we obtain measurable functions $\alpha^k \colon \Omega^k \to \RR^{m+1}$ which form a process over $\Omega$
\[
\alpha^0,\dots,\alpha^{n-1}.
\]
We remark that \eqref{E:L:linear programming:vertex} are the inequalities in \eqref{E:L:linear programming:equivalent-constraints} corresponding to $\omega=\rho_0,\dots,\rho_m$ with inequalities turned into equalities.

\noindent
{\em Claim 2:} Consider some $\omega^K \in \Omega^k$ where $0 \leq k \leq n-1$.
Then $\alpha^k(\omega^K)$ solves the inequalities \eqref{E:L:linear programming:equivalent-constraints} for all $\lambda \in \L \subseteq \Omega$.

\noindent
{\em Proof:}
As in Claim 1, write $\Phi(\omega)$ for the left hand side of \eqref{E:L:linear programming:equivalent-constraints} and $Y(\omega)$ for the right.
The claim is that $\Phi(\lambda) \geq Y(\lambda)$ for all $\lambda \in \L$.
Assume this is false, namely $\Phi(\lambda) < Y(\lambda)$ for some $\lambda \in \L$.
Among all these $\lambda$'s choose one for which $j$ is maximal with $\rho_j \preceq \lambda$; see Definition \ref{D:vertices of cube L} and the discussion below it.
Clearly $j<m$ because by definition of $\alpha^k(\omega^K)$ we have $\Phi(\rho_i)=Y(\rho_i)$ for all $0 \leq i \leq m$ and because $\rho_m \in \L$ is maximal.
Set $\lambda'=\lambda \vee \rho_{j+1}$.
By the choice of $\lambda$ we get $\lambda \wedge \rho_{j+1}=\rho_j$.
Since $S^{k+1}_i(\omega^K, \omega)=S_i^k(\omega^K) \cdot \Psi^{k+1}_i(\omega)$ and since the assignment $\omega \mapsto \Psi^{k+1}_i(\omega^K,\omega)$ is an affine function on $\Omega$, it follows that $\Phi \colon \Omega \to \RR$ is affine.
Therefore
\[
\Phi(\lambda')+\Phi(\rho_j) = \Phi(\lambda)+\Phi(\rho_{j+1}).
\]
The assumption on $X^{k+1}$ implies that $Y|_\L$ is supermodular, hence
\[
Y(\lambda')+Y(\rho_j) \geq Y(\lambda)+Y(\rho_{j+1}).
\]
Subtracting these inequalities, keeping in mind that by construction $\Phi(\rho_i)=Y(\rho_i)$, we get
\[
\Phi(\lambda')-Y(\lambda') \leq \Phi(\lambda)-Y(\lambda)<0.
\]
Therefore $\Phi(\lambda')<Y(\lambda')$ and $\rho_{j+1} \preceq \lambda'$.
This contradicts the maximality of $j$.
\hfil q.e.d

\noindent
{\em Claim 3:} $\alpha^k(\omega^K)$ solves the inequalities \eqref{E:L:linear programming:equivalent-constraints} for all $\omega \in \Omega$.

\noindent
{\em Proof:}
Since $\Omega$ is the convex hull of $\L$, any $\omega \in \Omega$ is a convex combination $\omega = \sum_{\lambda \in \L} t_\lambda \cdot \lambda$.
The assumption on $X^{k+1}$ implies that $Y \colon \Omega \to \RR$ is convex.
Together with Claim 2 and since $\Phi$ is affine
\[
\Phi(\omega) = 
\Phi(\sum_{\lambda \in \L} t_\lambda \lambda) =
\sum_{\lambda \in \L} t_\lambda \Phi(\lambda) \geq
\sum_{\lambda \in \L} t_\lambda Y(\lambda) \geq 
Y(\sum_{\lambda \in \L} t_\lambda \lambda) =
Y(\omega).
\]
\hfill q.e.d

\noindent
{\em Claim 4:} $V_\alpha^k=X^k$.

\noindent
{\em Proof:}
Denote $q_j=q(\rho_j)$.
Consider some $\omega^K \in \Omega^k$.
Equation \eqref{E:L:linear programming:vertex} defining $\alpha^k(\omega^K)$ yields
\begin{align*}
r^{k}(\omega^K) \cdot X^k(\omega^K) &=
E_q(\omega \mapsto X^{k+1}(\omega^K, \omega)) 
\\
&=
\sum_{j=0}^m q_j X^{k+1}(\omega^K,\rho_j) 
\\
& =
\sum_{j=0}^m \sum_{i=0}^m \alpha^k_i(\omega^K) S_i^k(\omega^K) \cdot q_j \Psi^{k+1}_i(\omega^K,\rho_j) 
\\
&=
\sum_{i=0}^m \alpha^k_i(\omega^K) S_i^k(\omega^K) \cdot E_q(\omega \mapsto \Psi^{k+1}_i(\omega^K,\omega)) 
\\
&=
r^{k}(\omega^K) \cdot V_\alpha^k(\omega^K).
\end{align*}
Since $r^k>0$ it follows that $V_\alpha^k(\omega^K)=X^k(\omega^K)$.
\hfill
q.e.d

Claim 3 implies that $\alpha^0,\dots,\alpha^{n-1}$ solve all the inequalities \eqref{E:L:linear programming:equivalent-constraints} and hence it solves the constraints  \eqref{E:L:linear programming:constraint}.
Claims 1 and 4 imply that $V_\alpha^k \leq V_\beta^k$ for all $k$ and all $\beta^0,\dots,\beta^{n-1}$ that satisfy \eqref{E:L:linear programming:constraint}.
To complete the proof of the lemma it only remains to prove:

\noindent
{\em Claim 5:} The system of equations \eqref{E:L:linear programming:vertex} has a unique solution given by the matrices $Q,T,M'$ as in the statement of the lemma.

{\em Proof:} 
Consider $0 \leq k \leq n-1$ and $\omega^K \in \Omega^{k}$.
For $0 \leq i,j \leq m$ set $\chi_i^{k+1}(j)=\Psi^{k+1}_i(\omega^K,\rho_j)$.
Notice that by the hypotheses 
\[
\chi_0^{k+1}(j)=r^{k}(\omega^K).
\]
For $1 \leq i \leq m$ observe that $L_i^{k+1}(\rho_j) = 1$ if $i \leq j$ and $L_i^{k+1}(\rho_j) = 0$ if $i>j$.
Since $\Psi^{k+1}=d_i^k+\Delta_i^k(\omega^K) L_i^{k+1}$ we get 
\[
\chi_i^k(j)= \left\{
\begin{array}{ll}
d_i^k(\omega^K)+\Delta_i^k(\omega^K)  & i \leq j \\
d_i^k(\omega^K)             & i> j
\end{array}\right.
\]
Since $S_i^{k+1}(\omega^K,\rho_j)=S_i^k(\omega^K) \cdot \Psi^{k+1}_i(\omega^K,\rho_j)$, the matrix representing the system \eqref{E:L:linear programming:vertex} is
\begin{multline*}
M = 
\begin{bmatrix}
S_0^k \chi_0^{k+1}(0) & S_1^k \chi_1^{k+1}(0) & \cdots & S_m^k \chi_m^{k+1}(0) \\
S_0^k \chi_0^{k+1}(1) & S_1^k \chi_1^{k+1}(1) & \cdots & S_m^k \chi_m^{k+1}(1) \\
\vdots              &                     &        &         \vdots \\
S_0^k \chi_0^{k+1}(m) & S_1^k \chi_1^{k+1}(m) & \cdots & S_m^k \chi_m^{k+1}(m)  
\end{bmatrix}
=
\\
\underbrace{
\begin{bmatrix}
1      & \chi_1^{k+1}(0) & \cdots & \chi_m^{k+1}(0) \\
1      & \chi_1^{k+1}(1) & \cdots & \chi_m^{k+1}(1) \\
\vdots &                &        &         \vdots \\
1      & \chi_1^{k+1}(m) & \cdots & \chi_m^{k+1}(m)  
\end{bmatrix}
}_{M''}
\cdot
\underbrace{
\begin{bmatrix}
r^{k} S_0^k    &   \\
               & S_1^k \\
               &      & \ddots \\
               &      &        & S_m^k 
\end{bmatrix}
}_{T}
\end{multline*}
with all entries evaluated at $\omega^K$.
Thus, $M,M'',T$ are functions $\Omega^k \to \OP{Mat}_{(m+1)\times(m+1)}(\RR)$.
Oserve that for all $1 \leq i \leq m$ and $1 \leq j \leq m$
\[
\chi_i^{k+1}(j)-\chi_i^{k+1}(j-1)=
\left\{
\begin{array}{ll}
\Delta_i^{k+1}(\omega^K)    & \text{if $i = j$} \\
0                & \text{if $i \neq j$}
\end{array}\right.
\]
It follows that
\[
Q \cdot M'' = 
\left[
\begin{array}{c|cccc}
1       & d_1^{k}      & d_2^{k}     & \cdots & d_m^{k} \\
\hline
0       & \Delta_1^{k} &              &        &          \\
0       &               & \Delta_2^{k} &                   \\
\vdots  &               &               & \ddots &           \\
0       &              &               &         & \Delta_m^{k}   
\end{array}
\right]
\]
with these matrices evaluated at $\omega^K$.
We denote the latter matrix by $M'$ and notice that it is invertible since $\Delta_i^{k} >0$.
In particular $M(\omega^K)$ is invertible for any $\omega^K \in \Omega^k$  so \eqref{E:L:linear programming:vertex} has a unique solution.
Note that  $M^{-1}=T^{-1} \cdot M'{}^{-1} \cdot Q$ is a measurable function $\Omega^n \to \OP{Mat}_{(m+1) \times (m+1)}(\RR)$ and the solution of \eqref{E:L:linear programming:vertex} is therefore the one given in the statement of the lemma.
\end{proof}

\begin{proof}[Proof of Theorem \ref{T:maximal hedging}]
We apply Lemma \ref{L:linear programming for hedging} with the following data.
The probability measure $q$ is the upper supervertex $q^* \colon \L \to \RR $ (Definition \ref{D:upper supermodular vertex}) extended to a probability measure on $\Omega$.
The processes $S^0,\dots,S^n$  and $\Psi^1,\dots,\Psi^n$ are the prices $S^k=(S_0^k,\dots,S_m^k)$ of the assets and their price jumps $\Psi^k=(\Psi_0^k,\dots,\Psi_m^k)$.
The process $r^0,\dots,r^{n-1}$ consists of the constant functions with value $R$.
The processes $d^0,\dots,d^{n-1}$ and $\Delta^0,\dots,\Delta^{n-1}$ have components $d_i^k$ and $\Delta_i^k$ ($1 \leq i \leq m$) where $d_i^k$ is constant with value $D_i$ and $\Delta_i^k$ is constant with value $U_i-D_i$.
The process $X^0,\dots,X^n$ are the upper bound of the option's $F$ price at time $k$, namely $X^k=\Gamma_{\max}(F,k)$.

We need to show that the conditions of the lemma are fulfilled.
First, $S_i^k>0$ for all $k$.
By construction of the model, $S_i^{k+1}=S_i^{k} \cdot \Psi_i^{k+1}$  and $\Psi_0^{k+1}=R=r^k$ and $\Psi_i^{k+1}=D_i+(U_i-D_i)L_i^k=d_i^k+\Delta_i^k L_i^{k+1} $ for all $0 \leq k \leq n-1$.
Also, 
\[
E_{q}(\omega \mapsto \Psi_i^{k+1}(\omega^K,\omega))=E_q(D_i+(U_i-D_i)\ell_i(\omega)) = D_i+(U_i-D_i)b_i=R=r^{k}(\omega^K)
\]
by construction of $q^*$ (Definition \ref{D:upper supermodular vertex}).
By Theorem \ref{T:rational values interval ends} 
\[
X^k(\omega^K)=R^{k-n} \sum_{J \in \C P_{n-k}(m)} q_J \cdot F(\omega^K,\rho_J).
\]
Since $F$ is fibrewise convex-supermodular by Proposition \ref{P:European basket is fibrewise convex-supermodular}, $X^k$ is a linear combination with non-negative coefficients of fibrewise convex-supermodular functions, hence it is one as well.
Finally, we check that
\begin{align*}
E_{q^*}(\omega \mapsto X^{k+1}(\omega^K,\omega)) & =
E_{q^*}\left( \omega \mapsto R^{k+1-n} \sum_{J \in \C P_{n-k-1}(m)} q_J \cdot F(\omega^K,\omega,\rho_J)\right) 
\\
& = 
R^{k+1-n} \sum_{j=0}^m q_j \sum_{J \in \C P_{n-k-1}(m)} q_J F(\omega^K,\rho_j,\rho_J)
\\
& =
R^{k+1-n} \sum_{J \in \C P_{n-k}(m)} q_J F(\omega^K,\rho_J) 
\\
& =
R \cdot X^k(\omega^K) 
\\
& = 
r^k(\omega^K) \cdot X^k(\omega^K).
\end{align*}
All the conditions of Lemma \ref{L:linear programming for hedging} are fulfilled and we obtain a process $\alpha^0,\dots,\alpha^{n-1}$ which minimises
\[
V_\alpha^k=\sum_{i=0}^m \alpha_i^k S_i^k
\]
subject to the requirement that for all $0 \leq k \leq n-1$
\[
\sum_{i=0}^m \alpha_i^k S_i^{k+1} \geq X^{k+1}=\Gamma_{\max}(F,k+1).
\]
Thus, $\alpha(k)=\alpha^k$ is a minimum-cost maximal hedging strategy as required with the formulas for its value given in the statement of the theorem.
It only remains to note that $Y_t(k)$ at the state of the world $\omega^K \in \Omega^k$ used in the statement of the theorem is precisely $X^{k+1}(\omega^K,\rho_t)$ because 
\begin{align*}
Y_t(k)(\omega^K) &= 
R^{k+1-n} \sum_{J \in \C P_{n-k-1}(m)} q_J\cdot (\sum_{i=0}^m c_i \chi_i(J) \chi_i(t) S_i^k(\omega^K)-K)^+
\\
&=
R^{k+1-n} \sum_{J \in \C P_{n-k-1}(m)} q_J\cdot (\sum_{i=0}^m c_i \chi_i(J) \Psi_i(\rho_t) S_i^k(\omega^K)-K)^+
\\
& =
R^{k+1-n} \sum_{J \in \C P_{n-k-1}(m)} q_J\cdot (\sum_{i=0}^m c_i \chi_i(J) S_i^k(\omega^K,\rho_t)-K)^+
\\
& = 
\Gamma_{\max}(F,k+1)(\omega^K,\rho_t)
\\
& =
X^{k+1}(\omega^K,\rho_t).
\end{align*}
\end{proof}

\bibliography{bibliography}
\bibliographystyle{plain}

\end{document}